\documentclass[
BCOR2pt, 
captions=nooneline, 
bibliography=totoc, 
numbers=noenddot, 
parskip=half, 
headings=normal, 
enabledeprecatedfontcommands,
abstracton 
]{scrartcl} 
\usepackage{etex}

\usepackage{wallpaper}
\usepackage{array}
\usepackage{color} 
\usepackage[T1]{fontenc} 
\usepackage[USenglish]{babel} 
\usepackage{graphicx} 
\usepackage{remreset} 
\usepackage[nouppercase]{scrpage2} 
\usepackage{amsmath} 
\usepackage{amssymb} 
\usepackage[misc]{ifsym}
\usepackage[
thmmarks, 
amsmath 
]{ntheorem} 
\usepackage{bm} 
\usepackage{mathabx}
\usepackage{todonotes}
\usepackage{bbm} 
\usepackage{enumitem} 
\usepackage[hang]{subfigure} 
\usepackage{wrapfig} 
\usepackage{tabularx} 
\usepackage{dcolumn} 
\usepackage{booktabs} 
\usepackage{listings} 
\usepackage{psfrag} 
\usepackage{rotating}
\usepackage{multirow}
\usepackage[olditem,oldenum]{paralist}
\usepackage{graphicx}  
\usepackage{epstopdf}
\usepackage{natbib}
\usepackage{placeins}
\usepackage{tikz}
\usetikzlibrary{matrix}
\usetikzlibrary{trees}
\usepackage {subfigure}
\usepackage{algorithm}
\usepackage[noend]{algpseudocode}
\definecolor{bluecolor}{RGB}{79,129,189}

\makeatletter
\newcommand\mytoday{\number\year-\ifcase\month\or 01\or 02\or 03\or 04\or 05\or 06\or 07\or 08\or 09\or 10\or 11\or 12\fi-\ifcase\day\or 01\or 02\or 03\or 04\or 05\or 06\or 07\or 08\or 09\or 10\or 11\or 12\or 13\or 14\or 15\or 16\or 17\or 18\or 19\or 20\or 21\or 22\or 23\or 24\or 25\or 26\or 27\or 28\or 29\or 30\or 31\fi} 
\makeatother
\setkomafont{sectioning}{\normalcolor\bfseries} 
\clearscrheadfoot 
\automark[section]{subsection} 
\setkomafont{captionlabel}{\bfseries} 
\renewcommand*{\subfigbottomskip}{-6pt} 
\newcolumntype{d}[2]{D{.}{.}{#1.#2}} 
\newcommand*{\abstractnoindent}{} 
\let\abstractnoindent\abstract
\renewcommand*{\abstract}{\let\quotation\quote\let\endquotation\endquote
\abstractnoindent}
\makeatletter
\renewcommand{\p@enumii}[1]{\theenumi(#1)}
\makeatother


\theoremstyle{break} 
\theoremheaderfont{\bfseries}
\theorembodyfont{\itshape}
\theoremseparator{}
\newtheorem{definition}{Definition}[section] 
\newtheorem{lemma}[definition]{Lemma}

\newtheorem{proposition}[definition]{Proposition}
\newtheorem{remark}[definition]{Remark}

\theoremstyle{nonumberbreak} 
\theoremsymbol{$\Box$}
\newtheorem{proof}{Proof}

\newcommand*{\IR}{\mathbb{R}}

\newcommand*{\IN}{\mathbb{N}}

\makeatletter
\def\BState{\State\hskip-\ALG@thistlm}
\makeatother

\newcommand*{\vol}{v}

\newcommand{\E}{{\mathbb{E}}}
\newcommand*{\parmu}{\boldsymbol{\mu}}

\newcommand*{\new}[1]{{\color{black}{#1}}}

\renewcommand*\d{\mathop{}\!\mathrm{d}}

\usepackage{epstopdf}
\usepackage{epsfig}

\begin{document}
\title{\textbf{\Huge{Calibration to American Options: Numerical Investigation of the
de--Ameri\-cani\-za\-tion Method}}}

\bigskip
\author{Olena Burkovska, Kathrin Glau, Maximilian Ga{\ss},\\Mirco Mahlstedt,  Wim Schoutens and Barbara Wohlmuth }

\maketitle

\begin{abstract}
American options are the reference instruments for the model calibration of a large and important class of single stocks. For this task, a fast and accurate pricing algorithm is indispensable. The literature mainly discusses pricing methods for American options that are based on Monte Carlo, tree and partial differential equation methods. We present an alternative approach that has become popular under the name \textit{de--Ameri\-cani\-za\-tion} in the financial industry. The method is easy to implement and enjoys fast run-times. Since it is based on ad hoc simplifications, however, theoretical results guaranteeing reliability are not available. To quantify the resulting \textit{methodological risk}, we empirically test the performance of the de--Ameri\-cani\-za\-tion method for calibration. We classify the scenarios in which de--Ameri\-cani\-za\-tion performs very well. However, we also identify the cases where de--Ameri\-cani\-za\-tion oversimplifies and can result in large errors.
	\end{abstract}
\minisec{Keywords}
	American options, calibration, binomial tree model, CEV model, Heston
model, L\'evy models, model reduction, variational inequalities\\
\section{Introduction}

The most frequently traded single stock options are of American type. In general, there exists a variety of (semi-)closed pricing formulas for European options. However, for American options, there hardly exist any closed pricing formulas, and the pricing under advanced models rely on computationally expensive numerical techniques such as the Monte Carlo simulation or partial (integro) differential methods.

The motivation behind the de--Ameri\-cani\-za\-tion methodology is to reduce the complexity as well as to lower the computational cost. In general, it is much faster to calibrate European options than American options. 

In the financial industry, the so-called de--Ameri\-cani\-za\-tion approach has become market standard: American option prices are transferred into European prices before the calibration process itself is started. This is usually done by applying a binomial tree. The method is also briefly mentioned by \cite{carr2010stock}, who describe how their implied volatility data, stemming from the provider OptionMetrics, is obtained by applying exactly this de-Americanization scheme.
Figure \ref{DeAm} illustrates the scheme of the de--Ameri\-cani\-za\-tion methodology.\\ 
\begin{figure}[h!]
\begin{center}
\begin{tikzpicture}
\node[text=white, fill=bluecolor, align=center, text width=4.5cm,font=\bf] (A) at (0,0) {Market Data:\\ American \\ Option Prices};
\node[text=white, fill=bluecolor, align=center, text width=4.5cm,font=\bf] (B) at (9.8,0) {de-Americanized \\ European \\ Option Prices};
\node[text=white, fill=bluecolor, align=center, text width=4.5cm,font=\bf] (C) at (0,-4) {Calibrated \\ Model Parameters};
\node[text=white, fill=bluecolor, align=center, text width=4.5cm,font=\bf] (D) at (9.8,-4) {Calibrated \\ Model Parameters};
\node[text=white, fill=gray, align=center, text width=1.8cm,font=\bf] (E) at (4.7,-0.65) {Binomial\\ Tree};
\node[text=black, align=left,font=\bf] (F) at (4.7,0.3) {Simplification};
\node[text=black, align=left,font=\bf] (G) at (1.25,-2) {Calibration};
\node[text=black, align=left,font=\bf] (G) at (11.05,-2) {Calibration};
\draw[->, bluecolor, line width=2mm] (B) to (D) ;
\draw[->, bluecolor, line width=2mm] (A) to (B);
\draw[->, bluecolor, line width=2mm] (A) to (C);

\end{tikzpicture}
\end{center}
\caption{De--Ameri\-cani\-za\-tion scheme: American option
prices are transferred into European prices before the calibration process itself is started. We investigate the effects of de--Ameri\-cani\-za\-tion by comparing the results to directly calibrating American options.}
\label{DeAm}
\end{figure}
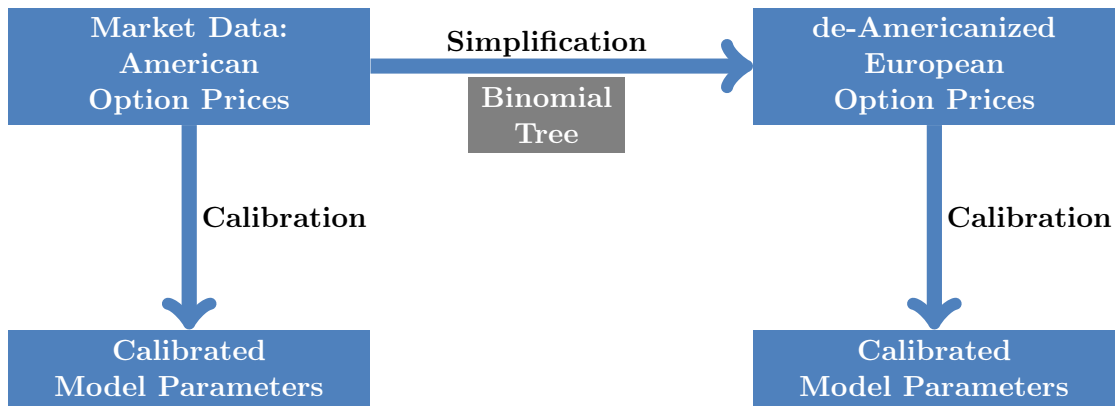\\
The de--Ameri\-cani\-za\-tion methodology enjoys three attractive features. It delivers fast run-times, is easy to implement and can flexibly be integrated into the pricing and calibration toolbox at hand. One downside is that no theoretical error control is available. Therefore, it is important to empirically investigate the accuracy, the performance and the resulting methodological risk of the method. In order to conduct a thourough investigation of these factors, we consider prominent models and identify relevant scenarios in which to perform extensive numerical tests. We explore the CEV model as an example of a local volatility model, the Heston model as a stochastic volatility model and the Merton model as a jump diffusion model. For all of these models, we implemented finite element solvers as benchmark method for pricing American options. The following questions serve as guidelines to specify decisive parameter settings within our studies.

\begin{enumerate}\label{Key_Questions}
\item Since American and European puts on non-dividend-paying underlyings coincide for zero interest rates, we analyze in particular the methodology for different interest rates.
\item Intuitively, with higher maturities, the early exercise feature of American options becomes more valuable and American and European option prices differ more significantly. Therefore, we investigate the following question: Does the accuracy of the de--Ameri\-cani\-za\-tion methodology depend on the maturity and do  de--Ameri\-cani\-za\-tion errors increase with increasing maturities?
\item In-the-money and out-of-the-money options play different roles. First, out-of-the-money options are preferred by practitioners for calibration since they are more liquidly traded, see for instance \cite{carr2010stock}. Second, in-the-money options are more likely to be exercised. How does the de--Ameri\-cani\-za\-tion methodology perform for out-of-the-money options and for in-the-money options?
\item The difference between American and European options is model-dependent. Intuitively, (higher) jump intensities lead to higher values of early exercise features. How does the de--Ameri\-cani\-za\-tion methodology perform for continuous models (CEV model and Heston model)? How does it perform for different jump intensities (Merton model)?
\end{enumerate}

Our investigation is organized as follows. First, we introduce the de--Ameri\-cani\-za\-tion methodology in Section 2. Then we briefly describe in Section 3 the models and the benchmark pricing methodology. Section 4 presents the numerical results:
The accuracy of the calibration procedure obviously hinges on the accuracy of the underlying pricing routine. We therefore first specify the de--Ameri\-cani\-za\-tion pricing routine and investigate its accuracy. Afterwards, we present the results of calibration to both synthetic data and market data. To conclude the numerical study, we present the effects of different calibration results on the pricing of exotic options. We summarize our findings in Section 5.

\section{De--Ameri\-cani\-za\-tion methodology}\label{Sec_2}
In this section, we give a precise and detailed description of the methodology. 
The de--Ameri\-cani\-za\-tion methodology is used to fit models to market data. The core idea of de--Ameri\-cani\-za\-tion is to transfer the available American option data into pseudo-European option prices prior to calibration. This significantly reduces the computational time as well as the complexity of the required pricing technique. Basically, de--Ameri\-cani\-za\-tion can be split into three parts. The first part consists in collecting the available market data. The currently observable price of the underlying $S_0$, interest rate $r$ and the available American option prices are collected. In the following, we will denote the American option price of the $i$-th observed option by $V_A^i$. We interpret the market data as the true option prices, thus we assume that the observed market prices $V_A^i$ can be interpreted as $V_A^i=\sup_{t\in[0,T_i]}E[e^{-rt} \widetilde{\mathcal{H}}_i(S_t)|\mathcal{F}_0],\quad i=1,\ldots,N$, where $\widetilde{\mathcal{H}}_i$ is the $i$-th payoff function, $T_i$ the maturity of the $i$-th option, and the expectations are taken under a risk-neutral measure, $\mathcal{F}$ is the natural filtration, and $N$ denotes the total number of options. Up to this point, no approximation has been used. 

The second step is the application of the binomial tree to create pseudo-European -- \textit{so called} de-Americanized -- prices based on the observed American market data. In this step, we look at each American option individually and find the price of the corresponding European option with the same strike and maturity. This European option is found by fitting a binomial tree to the American option. The binomial tree was introduced by \cite{crr_tree} as follows. Starting at $S_0$, at each time step and at each node, the underlying can either go up by a factor of $u$ or down by a factor of $\frac{1}{u}$ and
 the risk-neutral probability of an upward movement is given by 
\begin{equation}\label{risk_neutral_probability}
p=\frac{e^{r\Delta t}-\frac{1}{u}}{u-\frac{1}{u}}.
\end{equation} 
Once the tree is set up, options can be valuated by going backwards from each final node. Thus, path-dependent options can be evaluated easily. Since for each option $i$ the American option price $V_A^i$ is known, as well as $S_0$ and $r$, the only unknown parameter of the tree is the upward factor $u$. At this step, the upward factor $u_i^{*}$ is determined such that the price of the American option in the binomial tree matches the observed market price. Thus, denoting $\{0:\Delta t:T_i\}=\{0,\Delta t,2\Delta t,\ldots,T_i\}$, we have $\sup_{t\in\{0:\Delta t:T_i\}} E[e^{-rt}\widetilde{\mathcal{H}}_i(S_t^{u_i^{*}})|\mathcal{F}_0]=V_A^i$, where $S_t^{u_i^{*}}$ denotes the underlying process described by a binomial tree with upward factor $u_i^{*}$. The early exercise feature of American options is reflected in the fact that the the supremum is taken over all discrete time steps. A detailed description of pricing American options in a binomial tree model is given in \cite{van2006binomial}. Once $S_t^{u_i^{*}}$ is determined, the corresponding European option with the same strike and maturity as the American option is specified, $V_E^i=E[e^{-rT_i}\widetilde{\mathcal{H}}_i(S_{T_i}^{u_i^{*}})|\mathcal{F}_0]$. Note that fixing $u_i^{*}$ also implicitly determines the implied volatility.

Then, for each American option $V_A^i$, a corresponding European option $V_E^i$ has been found, and the actual model calibration can start. 
The goal is to fit a model $M$, depending on parameters $\mu \in\mathbb{R}^d$, where $d$ denotes the number of parameters in the model, to the European option prices $V_E^i,\ i=1,\ldots,N$. Denote by $S_{T_i}^{M(\mu)}$ the underlying process in model $M$ with parameters $\mu\in\mathbb{R}^d$. In the calibration, the parameter vector $\mu$ is determined by minimizing the objective function of the calibration. Algorithm \ref{Algorithm_DeAm} summarizes the de--Ameri\-cani\-za\-tion methodology in detail.
\begin{algorithm}
\caption{De--Ameri\-cani\-za\-tion methodology}\label{Algorithm_DeAm}
\begin{algorithmic}[1]
\Procedure{Collection of Observable Data}{}
\State $S_0$, $r$, 
\State $V_A^i=\sup_{t\in[0,T_i]}E[e^{-rt} \widetilde{\mathcal{H}}_i(S_t)|\mathcal{F}_0],\quad i=1,\ldots,N$
\EndProcedure
\Procedure{Application of the binomial tree to each option individually}{}
\State $\textbf{for}\ i=1:N$
\State $\quad\quad\text{Find } u_i^{*} \text{ such that} $
\State $\quad\quad\quad\quad\sup_{t\in\{0:\Delta t:T_i\}} E[e^{-rt}\widetilde{\mathcal{H}}_i(S_t^{u_i^{*}})|\mathcal{F}_0]=V_A^i$
\State $\quad\quad\text{Derive the corresponding European option price with } u_i^{*}$
\State $\quad\quad\quad\quad V_E^i=E[e^{-rT_i}\widetilde{\mathcal{H}}_i(S_{T_i}^{u_i^{*}})|\mathcal{F}_0]$
\State \textbf{end}
\EndProcedure
\Procedure{Calibration to European options}{}
\State Find $\mu$ such that the differences
\State $\quad\quad E[e^{-rT_i}\widetilde{\mathcal{H}}_i(S_{T_i}^{M(\mu)})]-V_E^i,\quad i=1,\ldots,N$
\State are minimized according to the objective function
\EndProcedure
\end{algorithmic}
\end{algorithm}

Regarding the uniqueness of the factor $u_i^{*}$ in the De--Ameri\-cani\-za\-tion methodology described in Algorithm \ref{Algorithm_DeAm}, we will first investigate the case of a European put option. Therefore, we interpret the risk-neutral probability in \eqref{risk_neutral_probability} as function of $u$, $p(u)=\frac{ue^{r\Delta t}-1}{u^2-1}$. At each node in the binomial tree we have a two-point distribution, that we call Bernoulli distribution \(X \sim QB(u)\), where the value $u$ is taken with probability $p(u)$ and the value $\frac{1}{u}$ is taken with probability $(1-p(u))$.

\begin{proposition}\label{Proposition_Convex_Order}
For \(i = 1, ...., n\), let \(X_i \sim QB(u)\) and \(Y_i \sim QB(u')\). If \(u \leq u'\), and $u$, $u'$ satisfy the conditions
\begin{enumerate}
\item $u,u'\ge e^{r\Delta t}+ \sqrt{e^{2r\Delta t}-1}$, and 
\item $u,u'\le\frac{(-e^{r\Delta t}k-1)-\sqrt{(e^{r\Delta t}k+1)^2-4e^{r\Delta t}k}}{2k}=\frac{1}{k}$ or $u,u'\ge\frac{(-e^{r\Delta t}k-1)+\sqrt{(e^{r\Delta t}k+1)^2-4e^{r\Delta t}k}}{2k}=\frac{e^{r\Delta t}}{k}$,
\end{enumerate}
then for any \(K \in \mathbb{R}\)
\begin{align*}
\E\left[\left(K - \prod_{i = 1}^n X_i\right)^+\right] \leq \E\left[\left(K - \prod_{i = 1}^n Y_i\right)^+\right].
\end{align*}
\end{proposition}

\begin{remark}
In the implementation of the tree, we set the time step size $\Delta t\approx0.0002$ and we use a simple bi-section approach as suggested by \cite{van2006binomial} to find $u^{*}$. Thus, given a market price $V_A$, starting with an upper bound $u_{\text{ub}}$ and a lower bound $u_{\text{lb}}$ satisfying the conditions in Proposition \ref{Proposition_Convex_Order} such that,
\begin{align*}
\sup_{t\in\{0:\Delta t:T_i\}} E[e^{-rt}\widetilde{\mathcal{H}}_i(S_t^{u_{\text{ub}}})|\mathcal{F}_0]&>V_A,\\
\sup_{t\in\{0:\Delta t:T_i\}} E[e^{-rt}\widetilde{\mathcal{H}}_i(S_t^{u_{\text{lb}}})|\mathcal{F}_0]&<V_A,
\end{align*}
the bi-section approach is started and the new candidate for $u^{*}$ is $\hat{u}=\frac{u_{\text{ub}}+u_{\text{lb}}}{2}$. When $\sup_{t\in\{0:\Delta t:T_i\}} E[e^{-rt}\widetilde{\mathcal{H}}_i(S_t^{\hat{u}})|\mathcal{F}_0]>V_A$,  we set $u_{\text{ub}}=\hat{u}$ for the next iteration,  otherwise $u_{\text{lb}}=\hat{u}$. As stopping criterion, we choose
\begin{align*}
\Big\vert \sup_{t\in\{0:\Delta t:T_i\}} E[e^{-rt}\widetilde{\mathcal{H}}_i(S_t^{\hat{u}})|\mathcal{F}_0]- V_A\Big\vert\le\varepsilon,
\end{align*}
and set $\varepsilon=10^{-5}$ in our implementation. In Proposition \ref{Proposition_Convex_Order} we have investigated the European put case and can deduce from the convex ordering that the put prices are monotonically increasing in $u$. For a strict order, the $u^{*}$-value is thus uniquely determined. In our case, the $u^{*}$-value can be determined uniquely as minimum of all $u$ values satisfying the stopping criterion. Moreover, this indicates that also the American put price in the binomial tree is increasing with increasing $u$. We validated this by numerical tests (not reported). This is in line with the recommendation in \cite{van2006binomial}. The only observed limitation is that the American put price can not given by an immediate exercise at the initial time. This is explained in detail in Remark \ref{Remark_Exercise_Region}.
\end{remark}

\section{Pricing Methodology}

In this section, we present a model formulation and numerical implementation
of the three investigated models (CEV, Heston, Merton). To
investigate the de--Ameri\-cani\-za\-tion methodology, we need to price the American
and European options. Our market data in the numerical study later on will be based on options on the Google stock (Ticker: GOOG). As Google does not pay dividends, we neglect dividend payments in our pricing methodology. Without dividend payments, for $r>0$, it holds in general that American calls coincide with European calls and only American puts have to be treated differently. The opposite is true for $r<0$, in which case American and European puts coincide and American and European calls have to be treated differently.

In general, for European options, there exists a
variety of fast pricing methodologies such as Fast Fourier Transform
(\cite{CarrMadan99,Raible}) or even closed-form solutions. The common approaches for pricing American options are P(I)DE methods using either the finite difference method (FDM) or a finite element method (FEM).
We choose FEM since it is typically more flexible. To solve the resulting variational inequalities for American options, we use the Projected SOR
Algorithm,~\cite{achdou}, \cite{seydel}, for the CEV and Merton models,  and the 
Primal Dual Active Set Strategy,~\cite{kunisch}, for the Heston model.


\subsection{Option Pricing Models}
We briefly present the models that we use for our study, namely the
constant elasticity of variance model (CEV), the stochastic volatility Heston
model, and the Merton model.

In all three of the models, the asset price dynamics $S_\tau$ are governed by a stochastic differential equation (SDE) of the form
\begin{subequations}\label{SDE1}
\begin{equation}
\d{S_\tau}=\ rS_\tau\d{\tau} + \sigma(S,\tau) S_\tau \d{W_\tau}+
S_{\tau-}\d{J_\tau},\qquad S_0=s\ge0,
\end{equation}
\begin{equation}
J_\tau =\ \sum_{i=0}^{N_\tau} Y_i,
\end{equation}
\end{subequations}
with $W_\tau$ a standard Wiener process, $r$ the risk-free interest rate and a volatility function $\sigma(S,\tau)$.
The jump part $(J_\tau)_{\tau\geq 0}$ is a compound Poisson process
with intensity $\lambda\geq 0$ and independent identically distributed jumps
$Y_i$, $i\in\IN$, that are independent of the Poisson process
$(N_\tau)_{\tau\geq 0}$. The Poisson process and the Wiener process are also independent.

As an example of a local volatility model, we begin by presenting the CEV model, which was introduced by \cite{Cox}. Here, the local volatility is assumed to be a deterministic function of the asset price for the process in~\eqref{SDE1},
$\sigma(S,\tau)=\sigma S_\tau^{\zeta-1}$, $0<\zeta<1$, $\sigma>0$ and $\lambda=0$.

As an example of a stochastic volatility model, we use the model proposed by \cite{heston}. In contrast to the CEV model, the stochastic volatility is driven by a second Brownian motion $\widetilde{W}_\tau$ whose correlation with $W_\tau$ is described by a correlation parameter
$\rho\in[-1,1]$, and the model is based on the dynamics of both the stock price~\eqref{SDE1}, with jump intensity $\lambda=0$, and the
variance $v_\tau$~\eqref{SDE2},
\begin{equation}
 dv_\tau=\kappa(\gamma-v_\tau)dt + \xi\sqrt{v_\tau}d\widetilde{W}_\tau,
 \label{SDE2}
\end{equation}
with $\sigma(S,\tau)=\sqrt{v_\tau}$, mean variance $\gamma>0$, rate of mean reversion
$\kappa>0$ and volatility of volatility $\xi>0$. Jumps are not included in either of the CEV or Heston models.\\

The Merton model includes jumps. The log-asset price process is not exclusively driven by a Brownian motion, but instead follows a jump-diffusion process. Thus, in the model of \cite{Merton1976}, the volatility of the asset process is still assumed to be constant, $\sigma(S,\tau) \equiv \sigma>0,\ \forall S>0, \forall \tau\geq 0.$ But being a jump diffusion model, the jump intensity $\lambda>0$ is positive and $N_t\sim \text{Poiss}(\lambda t)$. The jumps are taken to be independent normally distributed random variables, $Y_i \sim \mathcal{N}(\alpha,\beta^2)$ 
with expected jump size $\alpha\in\IR$ and standard deviation $\beta>0$. 

\subsection{Pricing P(I)DE}

Denote by
$t=T-\tau$ the time to
maturity $T$, $T<\infty$ and by $K$ the strike of an option. For the CEV model, we stay with the $S$ variable, $S\in(0,\infty)$, for the Heston and Merton model we work with the log-transformed stock variable $x:=\log\left(\frac{S}{K}\right)$,  $x\in(-\infty,\infty)$. 
In the following, we will denote an American or European call or put price by $P_{call/put}^{Am/Eu}$. For the CEV model we have $P_{call/put}^{Am/Eu}:(0,T)\times\mathbb{R}^+\rightarrow\mathbb{R}^+$ and for the Heston and Merton model we have $P_{call/put}^{Am/Eu}:(0,T)\times\mathbb{R}^n\rightarrow\mathbb{R}^+$ ($n=1\ (Merton),\ n=2\ (Heston)$).
The value of an option at $t=0$ is given by the payoff function
$\widetilde{\mathcal{H}}_{call/put}(\cdot)$,
$P_{call/put}(0)=P_0=\widetilde{\mathcal{H}}_{call/put}$ with 
$\widetilde{\mathcal{H}}_{call}(S):=(S-K)^+$ or
$\widetilde{\mathcal{H}}_{put}(S):=(K-S)^+$ in the CEV model and
$\widetilde{\mathcal{H}}_{call}(x):=(Ke^x-K)^+$,
($\widetilde{\mathcal{H}}_{put}(x):=(K-Ke^x)^+$) in the Heston and
Merton models.

Then, to find the value of the European option $P^{Eu}_{call/put}$, paying $P_0^{call/put}=\tilde{\mathcal{H}}_{call/put}(S)$ ($P_0^{call/put}=\widetilde{\mathcal{H}}_{call/put}(x)$) at
$t=0$
leads to solve the following initial boundary value problem
\begin{align}
\frac{\partial P^{Eu}_{call/put}}{\partial t}-\mathcal{L}^sP^{Eu}_{call/put}=0,\quad P^{Eu}_{call/put}(0)=P_0^{call/put},
\label{EO_PDE}
\end{align}
where the spatial partial (integro) differential operator
$\mathcal{L}^s$, $s=\{\rm CEV, H, M\}$ is
determined by the model used to price the option. 
For the CEV, Heston and Merton, it is given
by~\eqref{L_CEV},~\eqref{L_Heston} and~\eqref{L_Merton},
respectively.
\begin{subequations}
 \begin{align}
 \mathcal{L}^{\rm CEV}P_{call/put}^{Am/Eu}:&=\frac{\sigma S_t^{\zeta-1}}{2}S^2\frac{\partial^2 P_{call/put}^{Am/Eu}}{\partial S^2}+rS\frac{\partial P_{call/put}^{Am/Eu}}{\partial S}-rP_{call/put}^{Am/Eu} \label{L_CEV},\\
\mathcal{L}^{\rm H}P_{call/put}^{Am/Eu}:&=\frac{1}{2}v\frac{\partial^2 P_{call/put}^{Am/Eu}}{\partial
x^2}+\xi v\rho\frac{\partial^2
P_{call/put}^{Am/Eu}}{\partial v\partial x}+\frac{1}{2}\xi^2v\frac{\partial^2 P_{call/put}^{Am/Eu}}{\partial
v^2}\notag\\
&\quad\quad+\kappa(\gamma-v)\frac{\partial P_{call/put}^{Am/Eu}}{\partial
v}+\left(r-\frac{1}{2}v\right)\frac{\partial P_{call/put}^{Am/Eu}}{\partial
x}-rP_{call/put}^{Am/Eu},\label{L_Heston}\\
\mathcal{L}^{\rm M}P_{call/put}^{Am/Eu}:&= b\frac{P_{call/put}^{Am/Eu}}{\partial x}  + \frac{1}{2}\sigma^2\frac{\partial^2 P_{call/put}^{Am/Eu}}{\partial x^2}P_{call/put}^{Am/Eu}\notag\\
&\quad+ \int_\IR (P_{call/put}^{Am/Eu}(x+z)-P_{call/put}^{Am/Eu}(x)-\frac{P_{call/put}^{Am/Eu}(x)}{\partial x} z)F(\d{z})-rP_{call/put}^{Am/Eu} \label{L_Merton},
\end{align}\label{model_operators}
\end{subequations}
where, for the Merton model, the jump measure $F$ is
given by
	\begin{equation}
		F(\d{z}) = \frac{\lambda}{\sqrt{2\pi\beta^2}}\exp\left(-\frac{(z-\alpha)^2}{2\beta^2}\right)\d{z}
	\end{equation}
and the drift $b\in\IR$ is set to $b:=r -
\frac{1}{2}\sigma^2-\lambda\left(e^{\alpha+\frac{\beta^2}{2}}-1\right)$ due to
the no-arbitrage condition.

Due to its early exercise possibility, pricing an American option (e.g., put) results in
additional inequality constraints, and leads us to solve the following system of
inequalities
\begin{subequations}
 \begin{align}
\frac{\partial P_{call/put}^{Am}}{\partial t}-\mathcal{L}^sP_{call/put}^{Am}\geq 0, \quad  P_{call/put}^{Am}-P_0^{call/put}&\geq 0,\\
\left(\frac{\partial P_{call/put}^{Am}}{\partial
t}-\mathcal{L}^sP_{call/put}^{Am}\right)\cdot\left(P_{call/put}^{Am}-P_0^{call/put}\right)&=0.
\end{align}
\label{AO_put_model}
\end{subequations}

We denote the parameter vector by
$\parmu:=(\xi,\rho,\gamma,\kappa,r)\in\mathbb{R}^5$ for the Heston model,
$\parmu:=(\sigma, \zeta)\in\mathbb{R}^2$ for the CEV model and $\parmu:=(\sigma,
\alpha, \beta, \lambda)\in\IR^4$ for the Merton model. 
\new{Then the
problems~\eqref{EO_PDE},~\eqref{AO_put_model} are parametrized problems with
$\parmu\in\mathcal{P}$, where $\mathcal{P}\subset\mathbb{R}^d$ is a parameter
space. The solution can be written as $P=P(\parmu)$. In some cases, for notational convenience, we will omit the parameter-dependence of $P$ and related quantities.}

\subsection{Variational Formulation}

\subsubsection{Boundary Conditions}\label{sec_bc}
We tackle the non-homogeneous truncated
Dirichlet boundary conditions by means of the lift function $u_L(t)=g(t)$ onto the domain. For all models, we consider only Dirichlet- or Neumann-type
boundary conditions. For the European call in the Heston model, we specify
them as follows according to~\cite{winkler2001},
\begin{subequations}
\begin{align}
\Gamma_1: \ &\vol=\vol_{\min}\ \ \ \
&&P^{Eu}_{call}(t,\vol_{\min},x)=Ke^{x}\Phi(d_1)-Ke^{-rt}\Phi(d_2),\\
\Gamma_2: \ &\vol=\vol_{\max}\ \ \ \ &&P^{Eu}_{call}(t,\vol_{\max},x)=Ke^{x},
\end{align}\label{bc_log_eo}\end{subequations}
and we interpolate linearly on the boundaries  $\Gamma_3=\{x=x_{\min}\}$ and $\Gamma_4=\{x=x_{\max}\}$.
The cumulative
distribution function
$\Phi(\cdot)$ is defined in~\eqref{cum_nd} and $d_{1,2}=\frac{x+(r\pm\frac{\sigma^2}{2})t}{\sigma\sqrt{t}}$ with
$\sigma=\sqrt{\vol}$.

The boundary conditions for American put options in the Heston model are as
follows, according to~\cite{clarke} and \cite{during},
\begin{align*}
 P^{Am}_{put}(t,\vol,x)&=\widetilde{\mathcal{H}}_{put}(x), &&\text{on}\quad\Gamma_3\cup\Gamma_4,
&&&\\
 \frac{\partial P^{Am}_{put}}{\partial \vol}(t,\vol_{\min},x)&=0,
&&\text{on}\quad\Gamma_1,
 &&&\frac{\partial P^{Am}_{put}}{\partial \vol}(t,\vol_{\max},x)=0, \ \
\text{on}\quad\Gamma_2.
\end{align*}

For the CEV model, following \cite{seydel}, we applied the boundary conditions
\begin{align*}
P^{Am/Eu}_{call}(t,S_{min})&=0,\quad P^{Am/Eu}_{call}(t,S_{max})=S_{max}-e^{-rt}K, &\text{for call options},\\
P^{Eu}_{put}(t,S_{min})&=e^{-rt}K-S_{min},\quad P^{Eu}_{put}(t,S_{max})=0, &\text{for European put options,}\\
P^{Am}_{put}(t,S_{min})&=K-S_{min},\quad P^{Am}_{put}(t,S_{max})=0, &\text{for American put options}.
\end{align*}
In the Merton model, we subtract a function $\Psi$ from the original pricing PIDE that approximately matches the behavior of $P^\text{Merton}$ such that for all $t\in[0,T]$ we have $\widetilde{P}^\text{Merton} = P^\text{Merton}(t,x)-\Psi(t,x)\rightarrow 0$ for $x\rightarrow\pm \infty$.
We choose
	\begin{equation}
	\label{eq:MertonPsi}
		\begin{split}
			\Psi^\text{Am./Eu. call}(t,x) =&\ (Ke^x - Ke^{-rt})\Phi(x),\\
			\Psi^\text{Am. put}(t,x) =&\ (K-Ke^x)(1-\Phi(x)),
		\end{split}
	\end{equation}
for European call and put options, respectively, where $\Phi$ is the
cumulative distribution function of the normal distribution~\eqref{cum_nd},
\begin{equation}
 \Phi(x)=\frac{1}{\sqrt{2\pi}}\int_{-\infty}^x
e^{-\frac{1}{2}z^2}dz.
\label{cum_nd}
\end{equation}

The transformation of the Merton model obtained by subtracting an appropriately chosen
function $\Psi$ as introduced in~\eqref{eq:MertonPsi} results in zero boundary
conditions in space, $u(t,x_\text{min}) = u(t,x_\text{max}) = 0$ for all $t\in
[0,T]$.

\section{Numerical Study of the effects of de--Ameri\-cani\-za\-tion} 
Our main objective is to investigate the de--Ameri\-cani\-za\-tion methodology with respect to the previously stated questions 1-4 on page \pageref{Key_Questions}. But before we look at these questions and the calibration results in detail, we describe the discretization of our FEM pricers followed by an investigation of the effects of de--Ameri\-cani\-za\-tion on pricing. Then we switch to calibrating to synthetic data and, finally, to market data.

\subsection{Discretization}
We set up mesh sizes and time discretization in all three models such that the errors compared to benchmark solutions are roughly the same. In our test setting, we set $S_0=1$, $r=0.07$, $T=\{0.5,0.875, 1.25, 1.625,2\}$ and $K$ to 21 equally distributed values in $[0.5,1.5]$. For the discretization, we choose $[S_{min},S_{max}]=[0.01,2]$ for the CEV model, $[v_{min}, v_{max}]=[10^{-5},3]$ and  $[x_{min}, x_{max}]=[-5,5]$ for the Heston model and for the Merton model we set $[x_{min}, x_{max}]=[-5,5]$. We set $\mathcal{N}=1000$ for the CEV model, $\mathcal{N}=49\times 97=4753$ for the Heston model and $\mathcal{N}=192$ for the Merton model, as well as $\Delta t = 0.008$ for all models. For the CEV model, we choose $\sigma=0.15$ and $\zeta=0.75$ and as benchmark solution we implement the semi-closed-form solution of the CEV model for European put and call prices as shown in \cite{schroeder}. We use the semi-closed-form solution in~\cite{janek2011fx} for the Heston model as benchmark and as model parameters we use $\xi=0.1$, $\rho=-0.5$, $\gamma=0.05$, $\kappa=1.2$ and $v_0=0.05$.
In the Merton model, Fourier pricing is used as benchmark. The model is parametrized by setting $\sigma=0.2$, $\alpha=-0.1$, $\beta=0.1$ and $\lambda=3$. Summarizing the results, we observe that for all models, with the introduced discretization, the absolute error between the benchmark and the FEM solution is in the region of $10^{-3}$ to $10^{-4}$ and, thus, the pricers for all three models have comparable accuracy.
\subsection{Effects of de--Ameri\-cani\-za\-tion on Pricing}
\label{sec:Pricing}

First, we focus on pricing differences caused by de--Ameri\-cani\-za\-tion. Therefore, we compare the de--Americanized American prices with the derived European option prices in the following way. Starting with a set of model parameters, we price the American and European options. Then, the binomial tree is applied to translate the American option prices into de-Americanized pseudo-European prices. Subsequently, we compare the European and the pseudo-European \textit{so called} de-Americanized prices to identify the effects of the de--Ameri\-cani\-za\-tion methodology.

%

The advantage of this approach is that we can purely focus on de--Ameri\-cani\-za\-tion, decoupled from calibration issues. In order to do so, we define the following test set for the range of investigated options. Here, we focus on put options due to the fact that American and European calls coincide for non-dividend-paying underlyings.
\begin{align}
S_0&=1\notag\\
K&=0.80,0.85,0.90,0.95,1.00,1.05,1.10,1.15,1.20\notag\\
T&=\frac{1}{12},\frac{2}{12},\frac{3}{12},\frac{4}{12},\frac{6}{12},\frac{9}{12},\frac{12}{12},\frac{24}{12}\notag\\
r&=0, 0.01, 0.02, 0.05,0.07\label{DeAm_Pricing_Setting}
\end{align}
In each model, 5 parameter sets are investigated to cover the parameter range. These are summarized in Table \ref{Tab_Parameters_Pricing}.
\begin{table}[h!]
\centering
\begin{tabular}{ccc|ccccc|cccc}\toprule
      & \multicolumn{2}{c|}{CEV} & \multicolumn{5}{c|}{Heston}                    & \multicolumn{4}{c}{Merton} \\
      & $\sigma$   & $\zeta$ & $\xi$ & $\rho$ & $\gamma$ & $\kappa$ & $v_0$ & $\sigma$ & $\alpha$    & $\beta$    & $\lambda$ \\ \hline
$p_1$ & $0.2$      & $0.5$ & 0.10     & -0.20      & 0.07     & 0.1      & 0.07  & 0.20      &  -0.01      & 0.01       & 1 \\
$p_2$ & $0.275$    & $0.6$ & 0.25     & -0.50   & 0.10     & 0.4      & 0.10  & 0.15     &  -0.05      & 0.05       & 2 \\
$p_3$ & $0.35$     & $0.7$ & 0.40     & -0.50    & 0.15     & 0.6      & 0.15  & 0.20      &  -0.10       & 0.10        & 3 \\
$p_4$ & $0.425$    & $0.8$ & 0.55     & -0.45   & 0.20     & 1.2      & 0.20  & 0.10      &  -0.10       & 0.20        & 5 \\
$p_5$ & $0.5$      & $0.9$ & 0.70     & -0.80    & 0.30      & 1.4      & 0.30  & 0.10      &  -0.15      & 0.20        & 7     
\end{tabular}
\caption{Overview of the parameter sets used for the CEV, Heston and Merton models}
\label{Tab_Parameters_Pricing}
\end{table}
\paragraph{Motivation of the selected parameters for the CEV model}
The main feature of the CEV model is the elasticity of variance parameter $\zeta$, which is combined with the level of the underlying to obtain a local volatility, namely $\sigma(S,t)=\sigma S^{\zeta-1}$, reflecting the leverage effect. In our example, we investigate American puts and the option-holder benefits from decreasing asset prices. In general, increasing the volatility leads to increasing option prices, but especially compared to the classical Black-Scholes model we are interested in the question of how strongly the incorporated leverage effect influences the put prices and whether the differences between American and European puts can be captured by the binomial tree. Thus, our selection for $\zeta$ in $p_1$ is 0.5, which strongly differs from the Black-Scholes model, and then $\zeta$ is further increased up to 0.9 within the scenarios. Additionally, we increase the values of $\sigma$.

\paragraph{Motivation of parameter selection for the Heston model}
Similar to the CEV model, the (American) put prices increase with increasing volatility. We try to cover this effect by increasing the volatility of the volatility parameters and the correlation between the two stochastic processes. In general, 
for stocks, the correlation between the volatility and the underlying value is negative. Thus, in the de-Americanization study, we focus on negative correlation values $\rho$. Starting in $p_1$ with a relatively low volatility and a slightly negative correlation $\rho$, in $p_2$ to $p_4$ we increase the volatility of volatility parameter $\xi$, the mean reverting level $\gamma$ and the mean reverting speed $\kappa$, and also investigate higher negative values for the correlation $\rho$. In all scenarios, the initial volatility $v_0$ is set to match the mean reverting level, i.e., $v_0=\gamma$.

\paragraph{Motivation of parameter selection for the Merton model}
The Merton model is a jump diffusion model. Due to the early exercise feature of American options, the existence of jumps has a significant impact on American option prices. Consider for example an American put. Here, the option-holder benefits from decreasing asset prices. Consequently, when the possibility of negative jumps increases, the option price will increase as well. The jump intensity parameter $\lambda$ therefore plays a decisive role in this de-Americanization study. The analogous reasoning holds for the expected jump size parameter $\alpha$. In the upcoming numerical study, we try to incorporate these effects. The considered scenarios for the Merton model presented in Table~\ref{Tab_Parameters_Pricing} are chosen by this reasoning. Scenario $p_1$ describes a  Black-Scholes-like market with a rather low presence of jump occurrences. In $p_2$ and $p_3$, the jump feature appears more pronounced. Scenario $p_4$ and $p_5$ finally are encoded by rather jump-dominated parameter sets, which have an average number of $7$ jumps per year with large expected negative jump sizes that appear highly volatile.

\begin{remark}\label{Remark_Exercise_Region}
We price the put options in (\ref{DeAm_Pricing_Setting}) for the parameter sets shown in Table \ref{Tab_Parameters_Pricing}. For some parameters, especially for high interest rates combined with low volatility, it could occur that the price of an American put option equals exactly $K_i-S_0$, so that this American put option would be exercised immediately. In the following analysis, we excluded these cases because a unique European option price cannot be determined by applying the binomial tree. As illustrated in the following toy example in Figure \ref{Toy_Example}, there are several possible values for $u$ to replicate the American option price if the price of the American option is determined by immediately exercising it. In the example, a put option with strike $K=120$ is priced. Here, $u=1.04$ and $u=1.11$ are possible solutions. To avoid this, we consequently only consider American put options in our analyses when $P^{Am}_{put}>(K-S_0)^{+}\cdot (1+\delta)$. Thus, the American put option price exceeds the immediate exercise price by a factor of $\delta$. We set $\delta=1 \%$.  
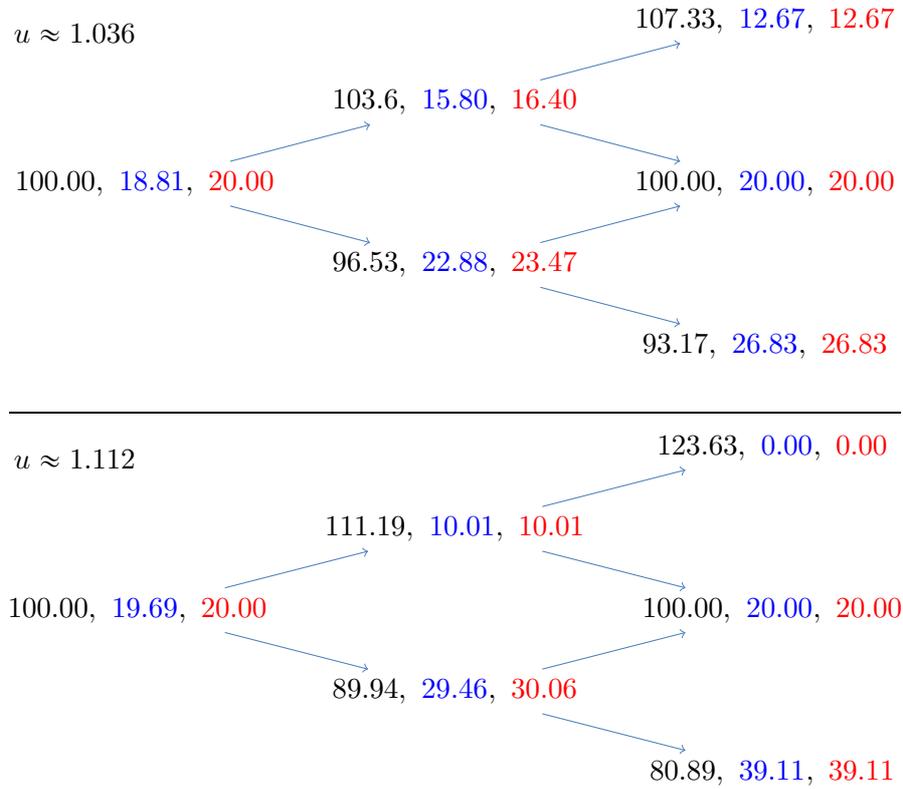
\begin{figure}
\centering

 \begin{tikzpicture}
    \matrix (tree) [%
      matrix of nodes,
      minimum size=0.25cm,
      column sep=0.5cm,
      row sep=0.5cm,
    ]
    {
          &   & $107.33,\ \textcolor{blue}{12.67},\ \textcolor{red}{12.67}$ \\
          & $103.6,\ \textcolor{blue}{15.80},\ \textcolor{red}{16.40}$ &   \\
      $100.00,\ \textcolor{blue}{18.81},\ \textcolor{red}{20.00}$ &   & $100.00,\ \textcolor{blue}{20.00},\ \textcolor{red}{20.00}$ \\
          & $96.53,\  \textcolor{blue}{22.88},\ \textcolor{red}{23.47}$ &   \\
          &   & $93.17, \ \textcolor{blue}{26.83},\ \textcolor{red}{26.83}$ \\
    };
    \draw[->, bluecolor] (tree-3-1) -- (tree-2-2) ;
    \draw[->, bluecolor] (tree-3-1) -- (tree-4-2) ;
    \draw[->, bluecolor] (tree-2-2) -- (tree-1-3) ;
    \draw[->, bluecolor] (tree-2-2) -- (tree-3-3) ;
    \draw[->, bluecolor] (tree-4-2) -- (tree-3-3) ;
    \draw[->, bluecolor] (tree-4-2) -- (tree-5-3) ;
    \node[text=black, align=left,font=\bf] (F) at (-5,2) {$u\approx1.036$};
  \end{tikzpicture}
  \rule{0.8\textwidth}{0.4pt}
 \begin{tikzpicture}
    \matrix (tree) [%
      matrix of nodes,
      minimum size=0.25cm,
      column sep=0.5cm,
      row sep=0.5cm,
    ]
    {
          &   & $123.63,\ \textcolor{blue}{0.00},\ \textcolor{red}{0.00}$ \\
          & $111.19,\ \textcolor{blue}{10.01},\ \textcolor{red}{10.01}$ &   \\
      $100.00,\ \textcolor{blue}{19.69},\ \textcolor{red}{20.00}$ &   & $100.00,\ \textcolor{blue}{20.00},\ \textcolor{red}{20.00}$ \\
          & $89.94,\  \textcolor{blue}{29.46},\ \textcolor{red}{30.06}$ &   \\
          &   & $80.89, \ \textcolor{blue}{39.11},\ \textcolor{red}{39.11}$ \\
    };
    \draw[->, bluecolor] (tree-3-1) -- (tree-2-2) ;
    \draw[->, bluecolor] (tree-3-1) -- (tree-4-2) ;
    \draw[->, bluecolor] (tree-2-2) -- (tree-1-3) ;
    \draw[->, bluecolor] (tree-2-2) -- (tree-3-3) ;
    \draw[->, bluecolor] (tree-4-2) -- (tree-3-3) ;
    \draw[->, bluecolor] (tree-4-2) -- (tree-5-3) ;
    \node[text=black, align=left,font=\bf] (F) at (-5,2) {$u\approx1.112$};
     
  \end{tikzpicture}
  \caption{Given an American put option price of 20 with $S_0=100$, $K=120$, $r=0.01$, i.e., an American put option in the exercise region, a unique tree cannot be found to replicate this option. In this example, we show two binomial trees for $u\approx1.036$ (top) as well as $u\approx1.112$ (bottom). In each tree, we show the value of the underlying (black), the European put price (blue) and the American put price (red) at each node. Both trees replicate the American option price of 20.00 but result in different European put prices: 18.81 and 19.69.}
  \label{Toy_Example}    
\end{figure}

\end{remark}
In Tables \ref{Tab_Pricing_CEV_avg_errors} - \ref{Tab_Pricing_CEV_max_prices} (CEV model), Tables \ref{Tab_Pricing_Heston_avg_errors} - \ref{Tab_Pricing_Heston_max_prices} (Heston model) and Tables \ref{Tab_Pricing_Merton_avg_errors} - \ref{Tab_Pricing_Merton_max_prices} (Merton model) in the appendix, we show in the appendix the pricing effects for the synthetic prices in \eqref{DeAm_Pricing_Setting}. For each scenario $p_i,i=1,\ldots,5$, we present the average difference between the de-Americanized prices and the European prices for each maturity and each strike and accordingly show the maximal European price in this maturity to reflect the issue stated in Remark \ref{Remark_Exercise_Region}. Similar studies have been done for the maximal error at each strike and maturity and confirm the findings based on the average error presented in the following.
In Figure \ref{Figure_Pricing_CEV}, we highlight the results for scenario $p_5$ in the CEV model to illustrate the effects of de-Americanization in several interest rate environments for different maturities or different strikes. For $p_5$ in the Heston model and $p_5$ in the Merton model, the results are shown in Figure \ref{Figure_Pricing_Heston} and Figure \ref{Figure_Pricing_Merton}, respectively. All of these figures clearly highlight the case $r=0$ as having hardly any de-Americanization effects (Heston and Merton) or at least fewer such effects (CEV).
\begin{figure} 
  \begin{minipage}{0.49\textwidth} 
     \centering 
     \includegraphics[width=\textwidth]{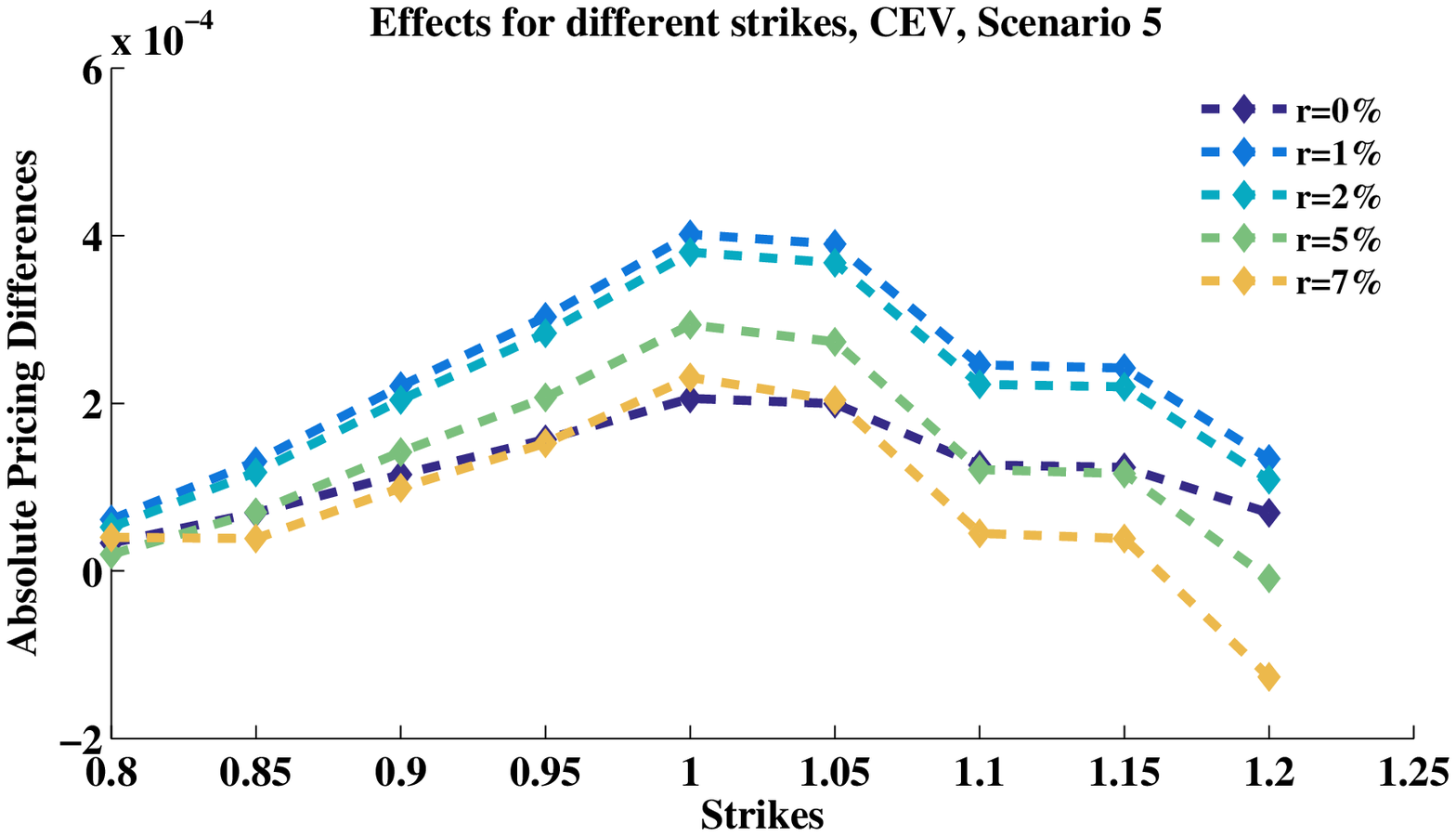} 
 \end{minipage} 
  \begin{minipage}{0.49\textwidth} 
     \centering 
     \includegraphics[width=\textwidth]{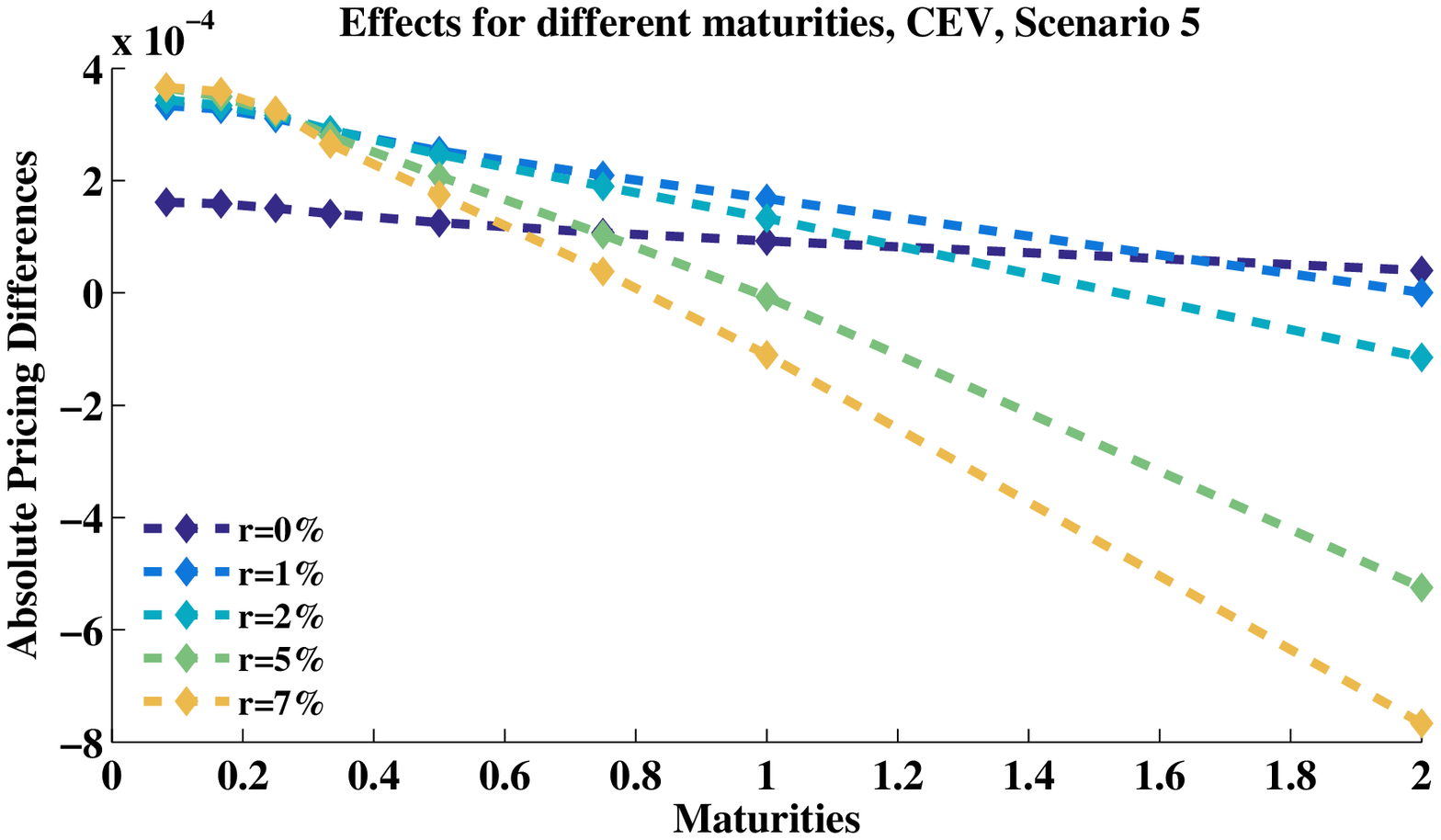} 
  \end{minipage} 
  \caption{De--Ameri\-cani\-za\-tion effects on pricing put options in the CEV model. As an example, the results are shown for $p_5$ for the average error between the de-Americanized and the European prices for each strike (left) and each maturity (right).}
  \label{Figure_Pricing_CEV} 
\end{figure} 
\begin{figure} 
  \begin{minipage}{0.49\textwidth} 
     \centering 
     \includegraphics[width=\textwidth]{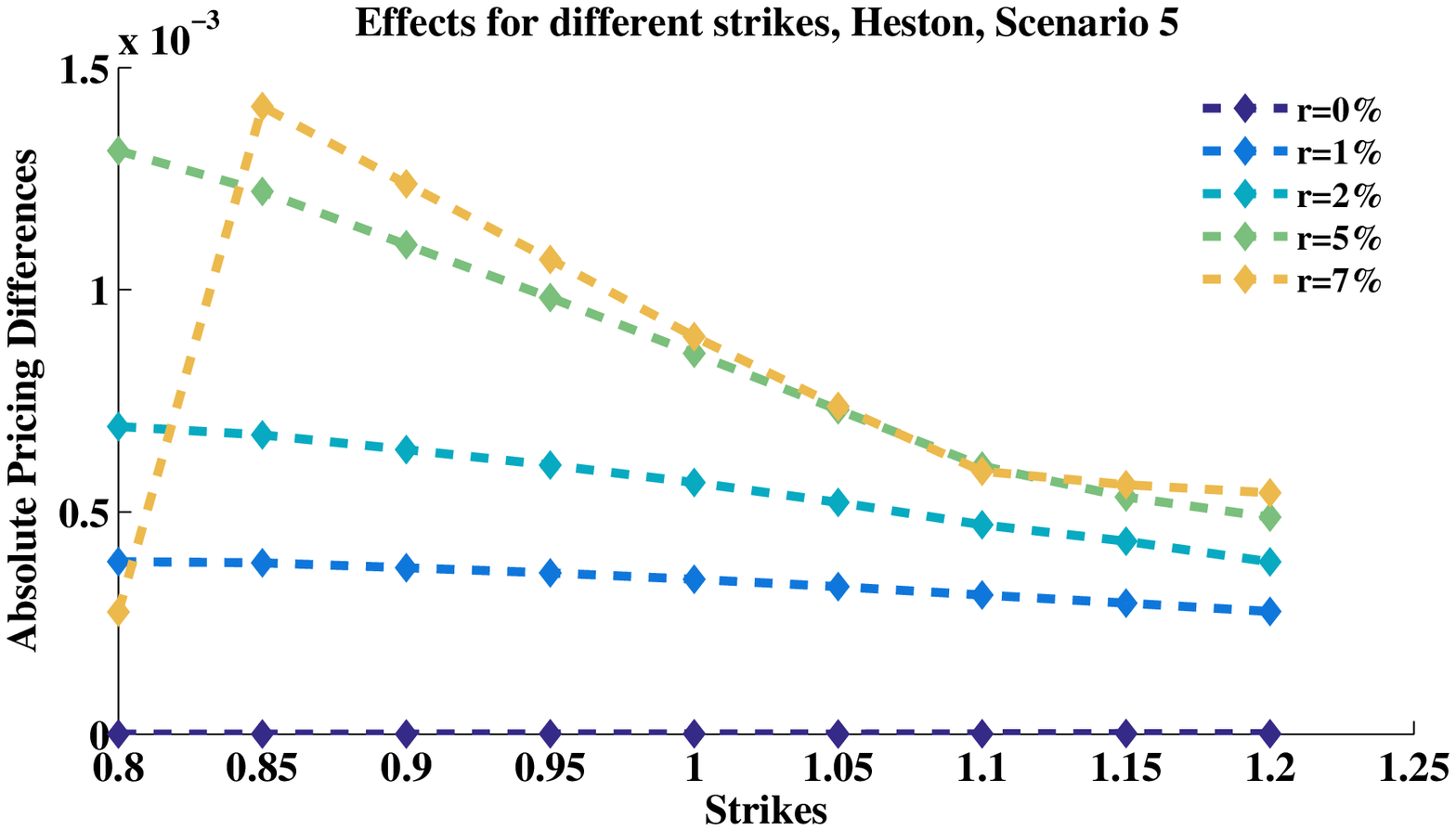} 
 \end{minipage} 
  \begin{minipage}{0.49\textwidth} 
     \centering 
     \includegraphics[width=\textwidth]{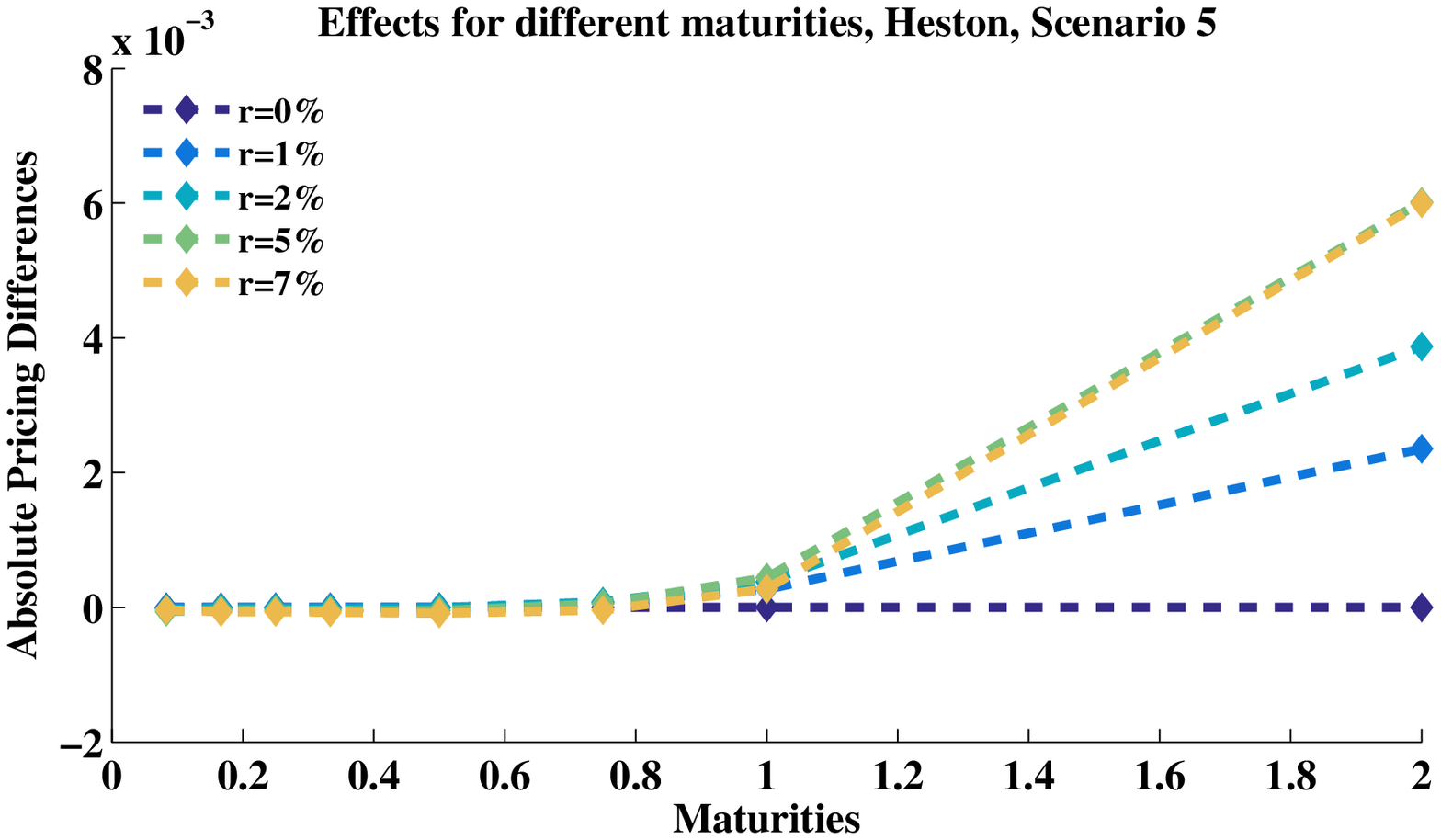} 
  \end{minipage} 
  \caption{De--Ameri\-cani\-za\-tion effects on pricing put options in the Heston model. As an example, the results are shown for $p_5$ for the average error between the de-Americanized and the European prices for each strike (left) and each maturity (right).}
  \label{Figure_Pricing_Heston} 
\end{figure} 

\begin{figure} 
  \begin{minipage}{0.49\textwidth} 
     \centering 
     \includegraphics[width=\textwidth]{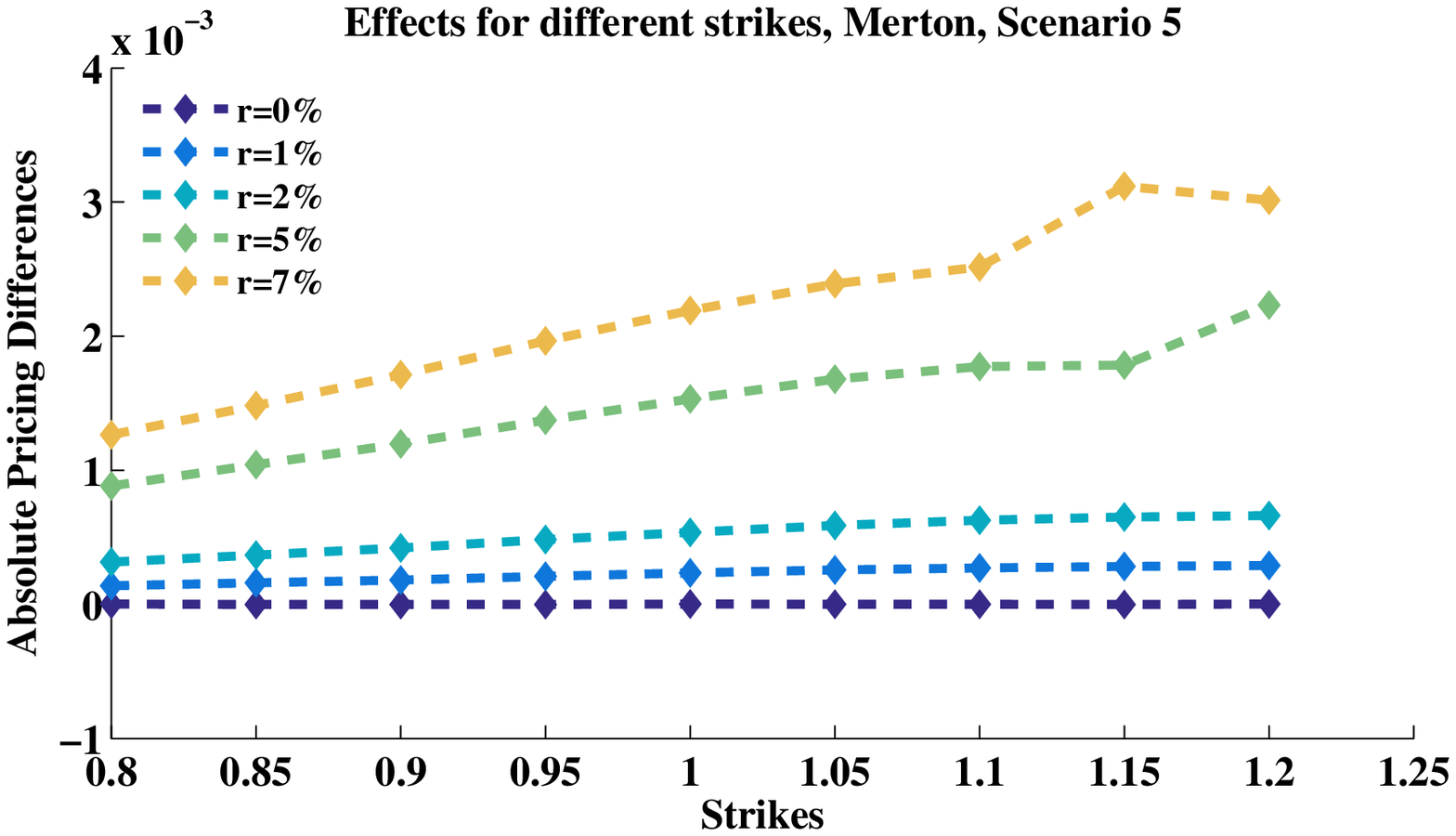} 
 \end{minipage} 
  \begin{minipage}{0.49\textwidth} 
     \centering 
     \includegraphics[width=\textwidth]{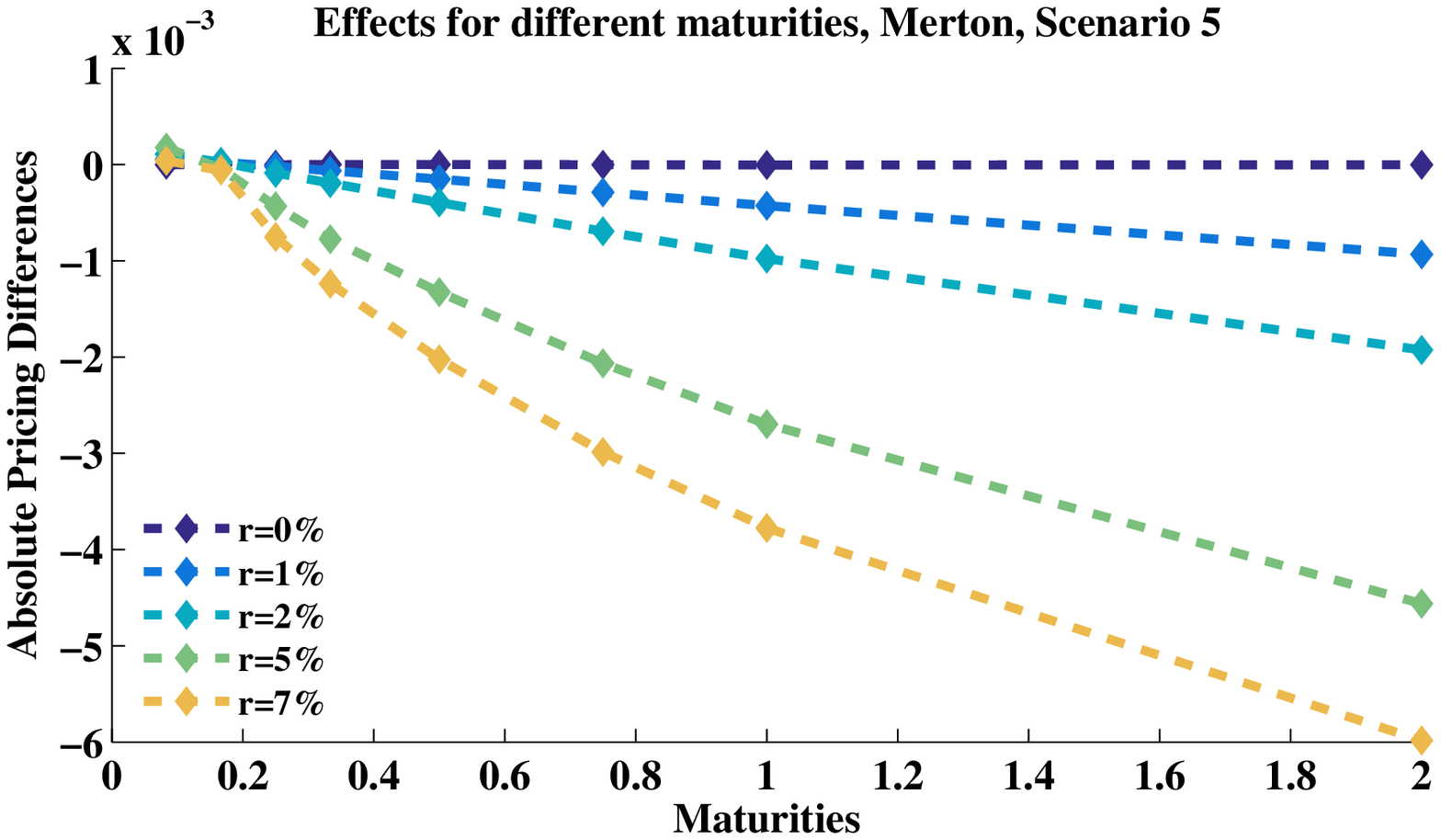} 
  \end{minipage} 
  \caption{De--Ameri\-cani\-za\-tion effects on pricing put options in the Merton model. As an example, the results are shown for $p_5$ for the average error between the de-Americanized and the European prices for each strike (left) and each maturity (right).}
  \label{Figure_Pricing_Merton} 
\end{figure} 

In general, for the CEV model, we observe that for short maturities the de--Americanized prices seem to overprice the European prices, whereas for longer maturities they seem to underprice the European options. We see that with increasing $\sigma$ and $\zeta$ parameters the maximal error increases and, overall, all parameter sets behave similarly. Focusing on the interest rate, we observe that for higher interest rates ($r=5\%$ and $r=7\%$) the average errors are higher or at least in a comparable region. Especially for higher interest rates, the maximal price has to be considered, because the higher the interest rate, the higher the probability that we did not consider some in-the-money options due to Remark \ref{Remark_Exercise_Region} and that the options with high prices are neglected in this setting. Thus we deduce that the error increases with increasing interest rates and that at high maturities the error increases for scenarios with higher volatility. For scenarios $p_1$, $p_2$ and $p_3$, we clearly observe that the effects of de-Americanization increase with increasing strikes. This means that for in-the-money options the de-Americanization effects tend to be stronger than for out-of-the-money options. This is consistent with the statements made by \cite{carr2010stock}. However, for higher interest rates, the average error seems to decrease with increasing strikes.
 This is due to the issue mentioned in Subsection \ref{Remark_Exercise_Region}. This effect occurred particularly strongly in the deep in-the-money region this effect occurred and the affected cases were neglected.\\

For the Heston model, we observe in general that the de--Ameri\-cani\-za\-tion error increases with increasing interest rates within each parameter setting. Additionally we see that for $r=0\%$ there is hardly any effect. By focusing on the scenarios with a higher volatility of volatility parameter ($p_4$ and $p_5$), we observe stronger de--Ameri\-cani\-za\-tion effects at short and long maturities ($T_1$ and $T_8$). For short maturities, the de--Americanized price is consistently lower than the corresponding European price throughout all scenarios,  whereas for high maturities the de--Americanized price is higher than the corresponding European price. To highlight the in-the-money and out-of-the-money issue in Figure \ref{Figure_Pricing_Heston}, note that in the Heston model the error is far smaller out-of-the-money  than deep in-the-money. However, the highest errors tend to occur in the at-the-money and slightly in-the-money regions.

For the Merton model, we observe similar, small effects for scenarios $p_1$ and $p_2$, i.e., the scenarios with low jump intensity, whereas for the scenarios with increasing jump intensity ($p_3$, $p_4$ and $p_5$) we observe stronger de--Ameri\-cani\-za\-tion effects, especially for increasing maturities and interest rates. Figure \ref{Figure_Pricing_Merton} additionally shows that for increasing strikes the effect of de--Ameri\-cani\-za\-tion increases slightly for lower interest rates and that for higher interest rates this error increases more strongly in the in-the-money region.

In addition to all of these de--Ameri\-cani\-za\-tion effects in absolute terms, we checked the magnitude of the relative error for the 1-year at-the-money put option, i.e.,  the absolute difference between the European and the de-Americanized price divided by the European price. In the CEV model, the average relative error for this option in all scenarios and interest rate settings was $0.1\%$ with a peak of $0.17\%$ at scenario $p_2$ with $r=1\%$. The average relative error for the Heston model was $0.17\%$ with a peak of $0.83\%$ in $p_2$ with $r=7\%$. In the Merton model, the average relative error of the at-the-money put option with maturity of one year was $0.18\%$ with a peak of $1.02\%$ at $p_5$ and $r=7\%$.

Summarizing the results, 
\begin{itemize}
\item de--Ameri\-cani\-za\-tion effects are sensitive to interest rate. The higher the interest rates, the higher the observable pricing differences,
\item de--Ameri\-cani\-za\-tion effects increase with increasing volatility and increasing maturities,
\item de--Ameri\-cani\-za\-tion effects tend to be stronger in-the-money, 
\item de--Ameri\-cani\-za\-tion effects increase with higher jump intensities.
\end{itemize}
Overall, in the settings mentioned above, we observe  a systematic effect caused by de--Ameri\-cani\-za\-tion. In the next step, we are interested in finding out whether these effects are also reflected in the calibration results.

\subsection{Effects of de--Ameri\-cani\-za\-tion on Calibration to Synthetic Data}
\label{sec:Calibration_Synthetic}

Here, we study the de-Americanization effect on synthetic American market data. To this effect, in a first step, we generate artificial market data using our FEM implementations of the three considered models. In a second step, we calibrate each model to the previously generated market data. This methodology allows us to disregard the noise affiliated with real market data is affiliated with and thus enables us to study the effect of de-Americanization exclusively.\\

Our artificial market data is specified as follows.
\begin{equation}
\begin{alignedat}{2}
S_0&=1\\
r&=7\%\\
T_1&=\frac{2}{12},\qquad &&K_1=\{0.95, 0.975, 1, 1.025, 1.05\},\\
T_2&=\frac{6}{12},\qquad &&K_2=\{0.9,0.925,K_1,1.075,1.1\},\\
T_3&=\frac{9}{12},\qquad &&K_3=\{0.85,0.875,K_2,1.125,1.15\},\\
T_4&=1,\qquad &&K_4=\{0.8,0.825,K_3,1.175,1.2\},\\
T_5&=2,\qquad &&K_5=\{0.75, 0.775, K_4,1.225,1.25\}.
\end{alignedat}
\label{eq:syndatasetting}
\end{equation}

As the data in~\eqref{eq:syndatasetting} shows, we consider a high-interest market and a set of maturities ranging from rather short-term American options with 2 months maturity to long-term American products with 2 years maturity. Each maturity $T_i$ is associated with a set of strikes $K_i$, $i\in\{1,\dots,5\}$. To analyze the effects of de--Ameri\-cani\-za\-tion on pricing, we price these options for the five parameter scenarios in Table \ref{Tab_Parameters_Pricing}. Regarding the calibration methodology, we have to make two choices. First, we have to decide which option types to include and, second, we need to determine the objective function.

Regarding the choice of options, we first consider only put options for the whole strike trajectory due to the fact that, in our setting of non-dividend paying underlyings, American and European calls coincide. Thus, we include in-the-money as well as out-of-the-money options. Second, motivated by the fact that the value of out-of-the-money options does not include any intrinsic 
value and is therefore supposed to better reflect the randomness of the market (as mentioned in \cite{carr2010stock}), we consider as a second approach that only includes out-of-the money puts and out-of-the money calls for the whole set of strikes and maturities. Consequently, in this second study, for each $i\in\{1,\dots,5\}$, we consider call option prices for maturities $T_i$ and strikes $k\in K_i$ with $k>1$ and put option prices for maturities $T_i$ and strikes $k\in K_i$ with $k<1$. At-the-money option data, i.e., options with strike $K=1$, is neglected.

Once the synthetic American market data has been generated, we create associated second synthetic market data by applying the de-Americanization routine using the binomial model.

Remark~\ref{Remark_Exercise_Region} and Figure~\ref{Toy_Example} describe situations in which the de-Americanization routine yields non-unique results. In the calibration to de-Americanized prices, we exclude options that cannot be de-Americanized uniquely as explained by the following remark.

\begin{remark}[Disregarding non-unique de--Americanized prices]
\label{rem:PutOptionsOutinCalib}
As outlined above, we artificially generate American market data for a calibration study on synthetic data. In a first step, we calibrate to the generated American prices directly. In a second step, we de--Americanize the option data and calibrate to the resulting quasi-European options. Here, we only consider option prices that admit a unique de--Americanized price. Consequently, all American put option prices that violate
	\begin{equation}
	\label{cond:DisregardAtCalib}
		P_{put}^{Am}>(K-S_0)^{+}\cdot (1+\delta),\qquad \text{ with } \delta=1 \%,
	\end{equation}
are not de--Americanized and thus are neglected in the second step.
\end{remark}

The second crucial assumption is the objective function. A variety of objective functions are proposed in the literature, e.g., the root mean square error, the average absolute error as a percentage of the mean price, the average absolute error, the average relative percentage error, absolute price differences, relative price differences, absolute implied volatilities, relative implied volatilities (see for example \cite{detlefsenhaerdle}, \cite{bauer}, \cite{fengler}, \cite{schoutens}).

We work directly with the observed prices and choose an objective function that considers prices, and due to the fact that the considered out-of-the-money option prices are rather small, we focus on absolute instead of relative differences. In the calibration, we take the absolute average squared error (aase) as the objective function and we minimize,
\begin{equation}
\text{aase}=\frac{1}{\#\text{options}}\sum_{\text{option}_k}\vert \text{Market price}_k-\text{Model price}_k\vert^2.\label{objectivefunction}
\end{equation}

The results of the calibration to synthetic data are summarized in Table \ref{Tab_Calibration_CEV} for the CEV and Merton models and in Table \ref{Tab_Calibration_Heston} for the Heston model for calibrating to put options and  calibrating to out-of-the-money options.
\begin{table}[h!]
\centering                                                                                                       
\begin{tabular}{ccl|ccc|ccccc}
&& &\multicolumn{3}{|c}{CEV}  &\multicolumn{5}{|c}{Merton}\\
                                                   &                                                 &      & $\sigma$           &$\zeta$                & aase & $\sigma$           &$\alpha$          & $\beta$           & $\lambda$                    & aase     \\
\multirow{5}{*}{$p_1$}                             &                                                 & true &   0.2      & 0.5   &---           &  0.20     &  -0.01  & 0.01      &  1    &  ---      \\
                                                   & \multirow{2}{*}{Put}                            & Am   & 0.1977      & 0.4962    & 7.74e-6   &  0.20    &  0.01   &    0.05   &   0.29  &    1.07e-10                \\
                                                   &                                                 & DeAm &   0.1894     &  0.4501   & 8.35e-5   &  0.20    &  -0.06   & 0.03      &  0.37    &8.37e-8                       \\
                                                   & \multirow{2}{*}{oom}                            & Am   &   0.1997     &0.4996     &  4.52e-6  &  0.20    &   0.00  &   0.05    &   0.32  &   1.35e-10                             \\
                                                   &                                                 & DeAm &  0.1793    &     0.9609&     2.75e-4   &   0.20  &  -0.02   &  0.04     &0.30      &     3.80e-9               \\
												   \hline
\multirow{5}{*}{$p_2$}                             &                                                 & true &   0.275     &  0.6  &  ---   & 0.15      &-0.05    & 0.05      &2      & ---     \\
                                                   & \multirow{2}{*}{Put}                            & Am   &    0.2740     &   0.6004  & 4.98e-7 &   0.15   &   -0.06  &   0.05    &   1.55  &  2.69e-11      \\
                                                   &                                                 & DeAm &    0.2607    &   0.7539  & 2.04e-6    &     0.14     &   -0.10  &  0.01     &  1.31    & 1.63e-7                  \\
                                                   & \multirow{2}{*}{oom}                            & Am   &    0.2736      & 0.5978    &  1.91e-6 &    0.16    & -0.10    & 0.04      &  0.66   & 1.91e-10                    \\
                                                   &                                                 & DeAm &    0.2484     &  0.5367   &  1.10e-5  &    0.15     & -0.11    &  0.03     & 0.74     &   2.22e-8          \\
												   \hline
\multirow{5}{*}{$p_3$}                             &                                                 & true & 0.35     &  0.7  &  ---   &   0.20  &-0.10    &0.10       &3      &        ---     \\
                                                   & \multirow{2}{*}{Put}                            & Am   &    0.3515     & 0.7576    & 4.37e-5  &    0.22  & -0.19    &    0.07   & 1.31    &         8.83e-10         \\
                                                   &                                                 & DeAm &     0.3272    &  0.8528   &   1.92e-4   &     0.16    & -0.05    &      0.11 &   5.04   &    3.12e-8    \\
                                                   & \multirow{2}{*}{oom}                            & Am   &     0.3476      &0.6984     &  1.00e-4   &     0.22   & -0.19    &   0.07    &1.32     & 5.47e-10    \\
                                                   &                                                 & DeAm &        0.3141     & 0.5527    &   5.99e-4   &      0.19   &  -0.17   &  0.09     &    1.88  &   3.33e-7             \\
												   \hline
\multirow{5}{*}{$p_4$}                             &                                                 & true &  0.425     & 0.8   &    ---      &  0.10       &-0.10    & 0.20      &5      &  ---      \\
                                                   & \multirow{2}{*}{Put}                            & Am   &    0.4258   &  0.7898   & 1.53e-6   &    0.10  &   -0.10  &    0.20   &   5.00  &    1.22e-13             \\
                                                   &                                                 & DeAm &     0.3942   &  0.8755   &  7.30e-6  &   0.10      & -0.10    & 0.20      & 5.00    &  6.29e-16              \\
                                                   & \multirow{2}{*}{oom}                            & Am   &   0.4262     & 0.7966    & 3.27e-6  &   0.09   & -0.09    &  0.21     &    4.90  &      2.19e-7  \\
                                                   &                                                 & DeAm &   0.3801     & 0.6009    & 1.96e-5  &   0.15     &  -0.14   &  0.22     &  3.63    &     6.53e-7           \\
												   \hline
\multirow{5}{*}{$p_5$}                             &                                                 & true &  0.5     & 0.9   &  ---   &  0.10  &-0.15    &0.20       & 7     &  ---        \\
                                                   & \multirow{2}{*}{Put}                            & Am   &    0.4982     &  0.9036   & 1.53e-6    &   0.10   & -0.15    &   0.20    & 7.00    & 8.96e-13        \\
                                                   &                                                 & DeAm &    0.4570      &  0.9192   &  1.02e-5   &   0.05    & -0.11    &    0.21   &   7.95   &   1.73e-7      \\
                                                   & \multirow{2}{*}{oom}                            & Am   &    0.4986     & 0.9036    &  4.02e-6  &  0.10      & -0.15    &   0.20    & 7.00    & 5.00e-13            \\
                                                   &                                                 & DeAm &    0.4430    & 0.6549    &   2.38e-5   &   0.05     &  -0.20   &       0.23&   5.08   &      1.67e-6             \\
												   \hline                      
\end{tabular}
\caption{Calibration results for calibrating to put options only and out-of-the-money options for the CEV model (left) and Merton model (right). Due to the effect of non-unique de-Americanization results, for the CEV model, some option prices have been neglected in the calibration to de-Americanized option data, as Remark~\ref{rem:PutOptionsOutinCalib} explains. In scenarios  $p_1$ to $p_5$, 5, 5, 10, 10 and 10 prices were excluded in the calibration to put options only. In scenarios  $p_1$ and $p_2$ of the Merton model, 5 prices have been excluded in the calibration to put options only.}
\label{Tab_Calibration_CEV}
\end{table}

\begin{table}[h!]
\centering                                                                                                       
\begin{tabular}{ccl|cccccc}
 &\multicolumn{8}{c}{Heston}  \\
                                                   &                                                 &      & $\xi$           &$\rho$          & $\gamma$           & $\kappa$          & $v_0$           & aase  \\
\multirow{5}{*}{$p_1$}                             &                                                 & true &    0.1     & -0.2   &  0.07     &  0.1    &  0.07    &  ---         \\
                                                   & \multirow{2}{*}{Put}                            & Am   &  0.1002   &-0.1999     &0.07       & 0.1026    &0.07      & 1.43e-13                   \\
                                                   &                                                 & DeAm &   0.1    & -0.4839    &   0.0651    &  0.5144    &  0.0695   &  1.50e-7                      \\
                                                   & \multirow{2}{*}{oom}                            & Am   &  0.1006   &  -0.1987   & 0.07      & 0.1049    & 0.07     & 4.73e-13                           \\
                                                   &                                                 & DeAm &  0.1  &  -0.1949   &   0.0665    & 0.2292     &  0.07   & 1.32e-8                   \\
												   \hline
\multirow{5}{*}{$p_2$}                             &                                                 & true &  0.25      &  -0.5  &  0.1     & 0.4     &  0.1    & ---    \\
                                                   & \multirow{2}{*}{Put}                            & Am   &   0.25    &   -0.5  &    0.1   &      0.4 &    0.1  & 6.47e-23      \\
                                                   &                                                 & DeAm &     0.2667  &   -0.5067  &   0.0978    &0.4374      &   0.0992  &  3.28e-8             \\
                                                   & \multirow{2}{*}{oom}                            & Am   &   0.25      & -0.5    &   0.1    &   0.4  &   0.1   &  2.99e-17                    \\
                                                   &                                                 & DeAm &     0.2199     &  -0.5   &    0.0885   & 0.1618     & 0.1    &   5.76e-9             \\
												   \hline
\multirow{5}{*}{$p_3$}                             &                                                 & true &  0.4     & -0.5   &  0.15     &   0.6   & 0.15     &  ---      \\
                                                   & \multirow{2}{*}{Put}                            & Am   &    0.4  &    -0.5 &       0.15&  0.6   &     0.15 &  7.15e-16           \\
                                                   &                                                 & DeAm &     0.4684 &  -0.437   &    0.1544   &    0.6806  & 0.1495    &  6.04e-9        \\
                                                   & \multirow{2}{*}{oom}                            & Am   &     0.4     &   -0.5  &   0.15    &    0.6 &    0.15  &7.03e-18      \\
                                                   &                                                 & DeAm &       0.3970  &   -0.5  &    0.1517   &  0.5413    &0.1494     &4.68e-9             \\
												   \hline
\multirow{5}{*}{$p_4$}                             &                                                 & true & 0.55  	  & -0.45   &  0.2     &  1.2    &  0.2    &  ---        \\
                                                   & \multirow{2}{*}{Put}                            & Am   &   0.55 		&    -0.45 &      0.2 &    1.2 &    0.2  &   1.44e-17             \\
                                                   &                                                 & DeAm &    0.5773	&     -0.4298&      0.2046 &  1.1975    &  0.1986   &  2.86e-9          \\
                                                   & \multirow{2}{*}{oom}                            & Am   &  0.55   & -0.45    &  0.2     &   1.2  &  0.2    &  8.41e-22    \\
                                                   &                                                 & DeAm &  0.5625   &  -0.4369   &  0.2035     & 1.2198     &  0.1988   & 4.50e-9                \\
												   \hline
\multirow{5}{*}{$p_5$}                             &                                                 & true &  0.7      &  -0.8  &    0.3   &   1.4   & 0.3     & ---       \\
                                                   & \multirow{2}{*}{Put}                            & Am   &    0.7   &     -0.8&   0.3    &  1.4   &   0.3   &  1.68e-17               \\
                                                   &                                                 & DeAm &     0.8504 &  -0.7057   &  0.3136     &  1.5832    & 0.2993    &  1.31e-8          \\
                                                   & \multirow{2}{*}{oom}                            & Am   &  0.7    &  -0.8   &    0.3   &   1.4  & 0.3     &  8.58e-23               \\
                                                   &                                                 & DeAm &   0.7763    &  -0.7602   & 0.3073      &1.7021      & 0.2979    & 2.34e-8            \\
												   \hline                      
\end{tabular}
\caption{Heston model: Calibration results for calibrating to put options only and out-of-the-money options.}
\label{Tab_Calibration_Heston}
\end{table}

Overall, we see that for the CEV model the parameters match well when calibrating to American options. When calibrating to de-Americanized prices however, the volatility parameter $\sigma$ is underestimated in most cases and this underestimation is counterbalanced by an overestimated $\zeta$-value.

Focusing on the Heston model, we observe that in every calibration to American options the parameters are matched better than in the corresponding calibration to de--Ameri\-cani\-zed data. We clearly see that the three parameters $\gamma$, $\kappa$ and $v_0$ are matched, but the remaining two parameters $\xi$ and $\rho$ show different results. When calibrating put options, as the volatility of volatility parameter $\xi$ increases,the calibrated de--Aameri\-cani\-zed parameter overestimates the true parameter. When calibrating out-of-the-money options, we observe that $\xi$ is underestimated for lower $\xi$ values and overestimated for higher $\xi$ values. Regarding $\rho$, as the volatility of volatility is increased, the de--Ameri\-cani\-zed parameter tends to underestimate the $\rho$ value. Later, by focusing on pricing exotic options, we will see whether these two contrary effects cancel each other or lead to different exotic option prices.\\
For the Merton model, we observe that when calibrating American options $\sigma$ is matched fairly accurately in most cases. For the other 3 parameters, we observe that whenever the jump intensity $\lambda$ is underestimated, the corresponding mean $\alpha$ and standard deviation $\beta$ of the jump are adjusted accordingly. Similar observations can be made for calibration to de--Ameri\-cani\-zed prices. Here, especially for calibrating  out-of-the-money values in $p_4$ and $p_5$, we observe that the $\sigma$ value is also not matched.

Summarizing the results, we observe that when calibrating de-Americanized synthetic data in a high-interest-rate environment for the continuous CEV and Heston models, the main parameters driving the volatility of the underlying, $\zeta$ and $\sigma$ (CEV) and $\xi$ and $\rho$ (Heston), are often not exactly matched. In these cases, the application of the binomial tree is not able to capture the volatility of the underlying exactly. For the jump model (Merton), we observe that due to the de-Americanization the jump intensity is (more strongly) mismatched than when directly calibrating to American options and in these cases the wrongly calibrated jump intensity parameter may be compensated by adjusting the other model parameters accordingly.
\subsection{Effects of de--Ameri\-cani\-za\-tion on Calibration to Market Data}
In this section, we investigate the effects of de--Ameri\-cani\-za\-tion by calibrating market data. The single stock of our choice is Google as an example of a non-dividend-paying stock. Table \ref{GoogleDaten} gives an overview of the processed data for the calibration procedure. In total we obtained a data set containing 482 options, with slightly more puts than calls. The risk-free interest rate for maturities of 1 month, 3 months, 6 months, 1 year and 2 years are taken from the U.S. Department of the Treasury\footnote{http://www.treasury.gov/resource-center/data-chart-center/interest-rates/Pages/TextView.aspx?data=yield} and have been linearly interpolated whenever necessary. 

\begin{table}[h]
\begin{center}
 \begin{tabular}{c c c c c}\toprule
  &Maturity&T&\# of options&r\\
  $T_1$&27.02.2015&0.07&47&0.0001\\
  $T_2$&20.03.2015&0.13&49&0.000129508\\
  $T_3$&17.04.2015&0.20&52&0.00017541\\
  $T_4$&19.06.2015&0.38&87&0.00046087\\
  $T_5$&18.09.2015&0.62&98&0.000955435\\
  $T_6$&15.01.2016&0.95&101&0.001602174\\
  $T_7$&20.01.2017&1.97&48&0.004786339\\
 \end{tabular}
 \caption{Processed Google option data for $t_0=\text{02.02.2015}$, $S_0=523.76$}
 \label{GoogleDaten}
 \end{center}
\end{table}

In order to structure the available data, we follow the methodology applied for the volatility index (VIX) by the Chicago board of exchange (\cite{vix}):
\begin{itemize}
\item Only out-of-the-money put and call options are used
\item The midpoint of the bid-ask spread for each option with strike $K_i$ is considered 
\item Only options with non-zero bid prices are considered
\item Once two puts with consecutive strike prices are found to have zero bid
prices, no puts with lower strikes are considered for inclusion (same for calls)
\end{itemize}
Basically, by this selection procedure, we only select out-of-the-money options that (due to non-zero bid prices) can be considered as liquid. In general, an option price consists of two components reflecting the time value and the intrinsic value of the option. By focusing on out-of-the-money options, the intrinsic value effects are mostly neglected and the highest option price will be at-the-money. Additionally, the highest market activity is in the at-the-money and slightly out-of-the-money region. 
The calibration results are summarized in Table \ref{Tab_Calibration_Google}. 

\begin{table}[]
\centering
\begin{tabular}{cl|cccccc}
                             &      & \multicolumn{6}{c}{CEV}                                                        \\
                             &      & \multicolumn{2}{c}{$\sigma$}    & \multicolumn{2}{c}{$\zeta$}    & \multicolumn{2}{c}{aase} \\
\multirow{2}{*}{Google Data} & Am   & \multicolumn{2}{c}{0.25} & \multicolumn{2}{c}{0.98} & \multicolumn{2}{c}{3}    \\
                             & DeAm & \multicolumn{2}{c}{0.25} & \multicolumn{2}{c}{0.97} & \multicolumn{2}{c}{3.32} \\ \hline
                             &      & \multicolumn{6}{c}{Heston}                                                     \\
                             &      & $\xi$           &$\rho$          & $\gamma$           & $\kappa$          & $v_0$           & aase       \\

\multirow{2}{*}{Google Data} & Am   & 0.2290      & -0.6854    & 0.0585      & 4.3186     & 0.0651      & 0.8464     \\
                             & DeAm & 0.2245      & -0.6941    & 0.0586      & 4.1433     & 0.0647      & 0.8319     \\ \hline
                             &      & \multicolumn{6}{c}{Merton}                                                     \\
                             &      & $\sigma$           & $\alpha$          & $\beta$           & $\lambda$          & \multicolumn{2}{c}{aase} \\
\multirow{2}{*}{Google Data} & Am   &   0.1936          &  -0.2000          &            0.2194 &     0.2935       &          \multicolumn{2}{c}{0.5813}            \\
                             & DeAm &   0.1935          &    -0.2000        &            0.2133 &     0.3014       &            \multicolumn{2}{c}{0.6035}             \\ \hline
\end{tabular}
\caption{Calibration results for calibrating to out-of-the-money put and call options combined.}
\label{Tab_Calibration_Google}
\end{table}

Here, we observe hardly any differences in the parameters. This is in line with our observations in Section \ref{sec:Pricing} for low-interest-rate environments. In these settings, American and European puts almost coincide and, thus, there will hardly be any difference in the prices and it is only natural that we observe very similar calibration results. Interestingly, the \textit{aase} value obtained by calibrating the Heston model is slightly lower when calibrating de-Americanized options than American options.

\subsection{Effects of de--Ameri\-cani\-za\-tion in pricing exotic options}

Plain vanilla options are traded liquidly in the market and are used to calibrate models. Financial institutions use these calibrated models to price more exotic products such as barrier and lookback options. In this subsection, we analyze which influences different calibration results have on the accuracy of exotic option prices. 

We analyze a down-and-out call option and a lookback option and hence translate differences in the calibrated model parameters into quantitative prices. The payoff $\widetilde{\mathcal{H}}_{DOC}(S(T))$ of a down-and-out call option with barrier $B$ is given by
\begin{align}
\widetilde{\mathcal{H}}_{DOC}(S(T))=(S(T)-K)^{+}\cdot\mathbbm{1}_{\min_{t\le T}S(t)\ge B}. 
\end{align}
In our setting, we set $S_0=100$, the barrier $B$ to 90\% of the initial underlying value and the strike $K$ to 105\% of the underlying value. For the lookback option, we choose the same strike and the payoff $\widetilde{\mathcal{H}}_{Lookback}(S(T))$ is
\begin{align}
\widetilde{\mathcal{H}}_{Lookback}(S(T))=(\bar{S}(T)-K)^{+},\quad with\ \bar{S}(T)=\max_{t\le T}S(t).
\end{align}

We price these two exotic options for the calibrated parameters in Tables \ref{Tab_Calibration_CEV}, \ref{Tab_Calibration_Heston} and \ref{Tab_Calibration_Google} via a standard Monte Carlo method with $10^6$ sample paths, 400 time steps per year and antithetic variates as variance reduction technique. The results are shown in the following Table \ref{Tab_Exotics}.
\begin{table}[h!]
\centering
\begin{tabular}{cclcc|cc|cc}
                                                   &                                                 &      & \multicolumn{2}{c}{CEV}                & \multicolumn{2}{c}{Heston}             & \multicolumn{2}{c}{Merton}                                 \\
                                                   &                                                 &      & \multicolumn{1}{c}{barrier} & lookback & \multicolumn{1}{c}{barrier} & lookback & \multicolumn{1}{c}{barrier} & \multicolumn{1}{c}{lookback} \\
\multirow{5}{*}{$p_1$}                             &                                                 & true & \multicolumn{1}{c}{9.93}        & 10.11         & \multicolumn{1}{c}{8.75}        &   8.80       & \multicolumn{1}{c}{3.86}        & \multicolumn{1}{c}{25.70}         \\
                                                   & \multirow{2}{*}{Put}                            & Am   &                    9.93         & 10.10         &               8.74              &     8.78     &          3.92                   &           25.86                   \\
                                                   &                                                 & DeAm &                9.93            &  10.01        &                 4.54            &    4.86      &              3.98               &                26.09             \\
                                                   & \multirow{2}{*}{oom}                            & Am   &               9.93              &  10.11        &                8.73             &   8.78       &                  3.93           &              25.86                \\
                                                   &                                                 & DeAm &             10.13                &  11.13        &                 8.27            &    8.37      &                      3.88       &                   25.77           \\
												   \hline
\multirow{5}{*}{$p_2$}                             &                                                 & true & \multicolumn{1}{c}{10.13}        &   11.14       & \multicolumn{1}{c}{2.28}        &   2.73       & \multicolumn{1}{c}{2.57}        & \multicolumn{1}{c}{22.97}         \\
                                                   & \multirow{2}{*}{Put}                            & Am   &          10.14                   &    11.14      &              2.28               &   2.73       &            2.49                 &                22.73              \\
                                                   &                                                 & DeAm &           11.47                  &    14.59      &                  2.07           &    2.54      &                2.56             &                     22.94         \\
                                                   & \multirow{2}{*}{oom}                            & Am   &           10.12                  &    11.10      &             2.27                &    2.73      &                    2.64         &                   23.24           \\
                                                   &                                                 & DeAm &           9.95                  &    10.40      &               4.32              &    4.57      &                        2.43     &                        22.75      \\
												   \hline
\multirow{5}{*}{$p_3$}                             &                                                 & true & \multicolumn{1}{c}{11.60}        &    14.93      & \multicolumn{1}{c}{1.15}        &   1.80       & \multicolumn{1}{c}{6.65}        & \multicolumn{1}{c}{37.35}         \\
                                                   & \multirow{2}{*}{Put}                            & Am   &         12.48                    &    17.83      &                1.15             &   1.80       &             6.78                &                      37.88        \\
                                                   &                                                 & DeAm &          13.53                   &    23.85      &                1.86             &    2.77      &                 6.50            &                           36.51   \\
                                                   & \multirow{2}{*}{oom}                            & Am   &          11.56                   &    14.81      &              1.14               &    1.80      &                     6.76        &                       37.85       \\
                                                   &                                                 & DeAm &           10.07                  &    10.91      &              1.40               &    2.08      &                         6.17    &                   37.63           \\
												   \hline
\multirow{5}{*}{$p_4$}                             &                                                 & true & \multicolumn{1}{c}{13.56}        &     24.08     & \multicolumn{1}{c}{0.83}        &   1.86       & \multicolumn{1}{c}{10.17}        & \multicolumn{1}{c}{54.99}         \\
                                                   & \multirow{2}{*}{Put}                            & Am   &    13.54                         &     24.00     &             0.83                &   1.86       &              10.14               &                55.00              \\
                                                   &                                                 & DeAm &    14.14                         &      30.87    &               1.04              &    2.21      &                   10.11          &                     54.98         \\
                                                   & \multirow{2}{*}{oom}                            & Am   &      13.51                       &      23.81    &            0.83                 &  1.86        &                10.18             &                          54.99    \\
                                                   &                                                 & DeAm &     10.60                        &     12.37     &            0.94                 &   2.05       &                     9.42        &                 55.54             \\
												   \hline
\multirow{5}{*}{$p_5$}                             &                                                 & true & \multicolumn{1}{c}{14.76}        &     43.40     & \multicolumn{1}{c}{0.02}        &     0.52     & \multicolumn{1}{c}{15.63}        & \multicolumn{1}{c}{76.15}         \\
                                                   & \multirow{2}{*}{Put}                            & Am   &           14.80                  &      43.99    &               0.02              &     0.52     &                 15.48            &                  75.69            \\
                                                   &                                                 & DeAm &         14.78                    &     43.50     &               0.03              &      0.67    &                      15.99       &                       75.97       \\
                                                   & \multirow{2}{*}{oom}                            & Am   &          14.74                   &      42.81    &             0.02                &     0.52     &                           15.58  &                      76.06        \\
                                                   &                                                 & DeAm &          11.68                   &      15.13    &               0.02              &     0.56     &                    13.98         &                           75.01   \\
												   \hline
\multicolumn{2}{c}{\multirow{2}{*}{\begin{tabular}[c]{@{}c@{}}Google \\ data\end{tabular}}} & Am   &                 14.21            &    32.10      &               0.69              &     1.51     &         4.03                    &28.91                              \\
\multicolumn{2}{c}{}                                                                                 & DeAm &           14.12                  &     30.70     &             0.67                &   1.47       &        4.03                     &        28.93                     
\end{tabular}
\caption{Overview of prices for barrier and lookback options}
\label{Tab_Exotics}
\end{table}

Overall, we observe a different picture in each of the three models. In $p_1$ and $p_2$ of the CEV model, the scenarios with relatively small volatility, we do not see any differences. Thus, in cases with small volatility and medium elasticity of variance $\zeta$, de--Ameri\-cani\-za\-tion seems to work. In the other scenarios, we observe that the calibration of de--Americanized prices leads to higher exotic option prices if we calibrate put options only and lower exotic option prices if we calibrate out-of-the-money options. Thus, the typically lower calibrated $\sigma$-value in combination with an increased $\zeta$-value obtained by calibrating de-Americanized options has this effect on the pricing of exotic options.

For the Heston model, we see that in the cases where the calibration of de--Americanized data led to different $\xi$ and $\rho$ values there are differences in the exotic option prices. More precisely, in all these cases, the corresponding barrier and lookback prices are too high. This means that de--Ameri\-cani\-za\-tion causes an systematic overpricing of exotic options.

For the Merton model, we see rather small differences for lookback options, but more interestingly, we observe differences for the down-and-out barrier option. This reflects the fact that the differently calibrated jump intensities and accordingly adjusted means and standard deviations of the jumps can buffer de--Ameri\-cani\-za\-tion effects over paths where the option cannot vanish like in the barrier option case.
 
In high-interest-rate environments, the de--Ameri\-cani\-za\-tion methodology leads to different exotic options prices in the CEV model when the volatility of the underlying is higher. When using only put options, the exotic option prices tend to be higher; when considering out-of-the-money options, the exotic option prices tend to be lower. In the Heston model, we observe a similar picture as in the CEV model, however here no general statement holds between higher and lower exotic option prices. Regarding the Merton model, the differences in the exotic option prices are more visible when considering the down-and-out barrier option.

\section{Conclusion}
In this paper, we investigate the de--Ameri\-cani\-za\-tion methodology by performing accuracy studies to compare the empirical results of this approach to those obtained by solving related variational inequalities for local volatility, stochastic volatility and jump diffusion models. 
On page \pageref{Key_Questions}, we propose key questions regarding the robustness of the de--Ameri\-cani\-za\-tion methodology with regard to changes in the (i) interest rates, (ii) maturities; (iii) in-the-money and out-of-the-money options and (iv) continuous and discontinuous models with increasing jump intensities.

First, focusing on pricing, we observe that de--Ameri\-cani\-za\-tion causes larger errors (i) for higher interest rates, (ii) for higher maturities, (iii) in the in-the-money region and (iv) for continuous models in scenarios with higher volatility and/or correlation, as well as in jump models for higher jump intensities. Second, we investigate model calibration to synthetic data for a specified set of maturities and strikes in a high-interest-rate environment. Numerically, we observe noticeable differences in the calibration results of the de--Ameri\-cani\-za\-tion methodology compared to the benchmark. For continuous models, the main difference lies in the resulting volatility parameters. For the jump model, the jump intensity is underestimated by the de--Ameri\-cani\-za\-tion methodology, especially in settings with high jump intensities, whereas the mean and the standard deviation of the jumps are overestimated. When calibrating Google market data, hardly any differences occur, which can be explained by the very low-interest-rate environment. This is in line with the results for question (i). 

In a final step, we investigate the effects of de--Ameri\-cani\-za\-tion in the model calibration on pricing exotic options. Here, exotic option prices play the role of a measure of the distance between differently calibrated model parameters. In most cases, we observe that exotic option prices are reasonably close  to the benchmark prices. However, we observe severe outliers for all investigated models. We find scenarios in which the exotic option prices differ by roughly 50\% in the CEV model ($p_4$) and the Heston model  ($p_1$) and by roughly 10\% in the Merton model ($p_5$), see Table \ref{Tab_Exotics}. Whereas in the CEV model and the Merton model the differences tend to be higher when calibrating to out-of-the-money options instead of only to put options, in the Heston model we have a mixed picture for different scenarios.

In a nutshell, the methodological risk of de--Ameri\-cani\-za\-tion critically depends on the interest rate environment.
For low-interest-environments, the errors caused by de--Ameri\-cani\-za\-tion are negligibly small and the de--Ameri\-cani\-za\-tion methodology can be employed when fast run-times are preferred.
For higher-interest-rate environments, however, de--Ameri\-cani\-za\-tion leads to uncontrollable outliers. Therefore, and since the de--Ameri\-cani\-za\-tion methodology does not provide an error control, we strongly recommend applying a pricing method in the calibration that is certified by error estimators.

We leave the inclusion of dividends for future research. The numerical results for the jump diffusion model and the sensitivity to interest rates indicate that discrete and continuous dividends may intensify the errors caused by the de--Ameri\-cani\-za\-tion method.

\section*{Acknowledgements}

We gratefully acknowledge valuable feedback from David Criens, Maximilian Mair and Jan Maruhn. Moreover, we thank the participants of the conferences Challenges in Derivatives Markets, March 30th - April 1st 2015, in Munich, the 12th German Probability and Statistic Days 2016, March 1st - 4th, in Bochum, and the Vienna Congress of Mathematical Finance 2016, September 12th - 14th, in Vienna as well as the participants of the KPMG Center of Excellence in Risk Management Research Day 2015, in Munich.

\section*{Funding}

This work was partly supported by: DFG grant WO671/11-1; International Research Training Group IGDK1754, funded by the German Research Foundation (DFG) and the Austrian research fund (FWF); KPMG Center of Excellence in Risk Management.

%

\bibliographystyle{rQUF}
\bibliography{references}

\appendix
\section{Proof of Proposition \ref{Proposition_Convex_Order}}
\subsection{Proof of Proposition \ref{Proposition_Convex_Order}}
\begin{proof}
Thanks to Lemma \ref{Lemma_Order_1Step}, for any \(i \leq n\) we have \(X_i \leq_{cx} Y_i\).
Hence, in view of \cite[Theorem 3.4.2]{muller2002comparison}, for each \(i \leq n\) there exists a probability space \((\Omega^i, \mathcal{F}^i, \mathbb{P}^i)\) which supports random variables \(x^i\) 
and \(y^i\) such that \(x^i \sim QB(u)\), \(y^i \sim QB(u')\), and 
\begin{align*}
\E^{\mathbb{P}^i} \left[ y^i |x^i\right] = x^i\quad \mathbb{P}^i\textup{-a.s.}
\end{align*}
Let us define 
\begin{align*}
\Omega := \bigtimes_{i = 1}^n \Omega^i,\quad \mathcal{F} := \bigotimes_{i = 1}^n \mathcal{F}^i, \quad \mathbb{P} := \bigotimes_{i = 1}^n \mathbb{P}^i,
\end{align*}
and extend \(x^i\) and \(y^i\) to \((\Omega, \mathcal{F}, \mathbb{P})\) by setting
\begin{align*}
x^i(\omega^1, ..., \omega^n) := x^i(\omega^i),\quad y^i(\omega^1, ..., \omega^n) := y^i(\omega^i).
\end{align*}
Now, it is easy to see that for measurable \(A, B \subset \mathbb{R}\)
\begin{equation}
\label{X}
\begin{split}
\mathbb{P}(x^i \in A, x^j \in B) &= \mathbb{P}(x^i \in A) \mathbb{P}(x^j \in B),\quad i \not = j,\\
\mathbb{P}(y^i \in A, y^j \in B) &= \mathbb{P}(y^i \in A) \mathbb{P}(y^j \in B), \quad i \not = j,
\end{split}
\end{equation}
and furthermore 
\begin{align}\label{Y}
\mathbb{P}(x^i \in A, y^j \in B) = \mathbb{P}(x^i \in A) \mathbb{P}(y^j \in B),\quad i \not = j.
\end{align}
From \eqref{X} we conclude that \(x := (x^1, ..., x^n) \sim (X^1, ..., X^n) =: X\), and \(y := (y^1, ...,  y^n) \sim (Y^1, ..., Y^n) =: Y\).
Moreover, in view of \eqref{Y}, we obtain
\begin{align*}
\E[ y| x] = 
\begin{pmatrix}
\E[y^1| x]\\
\vdots\\
\E[y^n|x]
\end{pmatrix}
=
\begin{pmatrix}
\E[y^1| x^1]\\
\vdots\\
\E[y^n|x^n]
\end{pmatrix}
= x\quad \mathbb{P}\textup{-a.s.}
\end{align*}
Therefore, by an application of \cite[Theorem 3.4.2]{muller2002comparison}, we obtain that \(X \leq_{cx} Y\).
Now let us note the following elementary fact: If \(f\colon U \to \mathbb{R}\) is concave, and \(g \colon \mathbb{R} \to \mathbb{R}\) is convex and decreasing, where \(U\subset\mathbb{R}^n\), then \(g \circ f\colon U \to \mathbb{R}\) is convex. An application of this fact with 
\begin{align*}
f(z_1, ..., z_n) := \prod_{i = 1}^n z_i,\ n\text{ even, }\quad g(z) := (K - z)^+,
\end{align*}
and the convex order \(X \leq_{cx} Y\) yields the claim. Note that by Lemma \ref{Lemma_Concave} $f$ is concave.
\end{proof}
\subsection{Additional Lemmata}

\begin{lemma}\label{Lemma_Order_1Step}
Focusing on one node in the binomial tree, let \(X \sim QB(u)\) and \(Y \sim QB(u')\) with $u'\ge u$. 
Let
\begin{enumerate}
\item $u,u'\ge e^{r\Delta t}+ \sqrt{e^{2r\Delta t}-1}$, and 
\item $u,u'\le\frac{(-e^{r\Delta t}k-1)-\sqrt{(e^{r\Delta t}k+1)^2-4e^{r\Delta t}k}}{2k}=\frac{1}{k}$ or $u,u'\ge\frac{(-e^{r\Delta t}k-1)+\sqrt{(e^{r\Delta t}k+1)^2-4e^{r\Delta t}k}}{2k}=\frac{e^{r\Delta t}}{k}$,
\end{enumerate}
be satisfied. Then the random variable $X$ is smaller than the random variable $Y$ with respect to the convex order, i.e. $X\le_{cx}Y$.
\end{lemma}
\begin{proof}
Following \cite[Theorem 1.5.3 and Theorem 1.5.7]{muller2002comparison} it suffices to show
\begin{enumerate}
\item $E[X]=E[Y]$
\item $E[(X-k)^+]\le E[(Y-k)^+]$
\end{enumerate}
Since $p$ as in \eqref{risk_neutral_probability} is set up as risk-neutral probability, it holds for any $u$ that $E[X]=e^{r\Delta t}$ and therewith, the first condition is satisfied. Given a random variable $X$ with a factor $u$ and a random variable $Y$ with factor $u'>u$, we distinguish regarding the second condition 5 cases.\\ 

\textit{Case 1:} $\frac{1}{u'}<\frac{1}{u}<u<u'<k$.\\

\noindent Obviously, in any case both options are out-of-the-money and $E[(X-k)^+]=0=E[(Y-k)^+]$.\\

\textit{Case 2:} $\frac{1}{u'}<\frac{1}{u}<u<k<u'$.\\

\noindent Here, $E[(X-k)^+]=0$ and $E[(Y-k)^+]=p(u')(u'-k)>0$, because only the second option is in-the-money in the up-case. Therewith, the second condition is satisfied.\\

\textit{Case 3:} $\frac{1}{u'}<\frac{1}{u}<k<u<u'$.\\

\noindent In this case both options are in-the-money in the according upward case.
\begin{align*}
E[(X-k)^+]-E[(Y-k)^+]&=p(u)(u-k)-p(u')(u'-k)\\
&=\frac{u^2e^{r\Delta t}-u}{u^2-1}-\frac{u'^2e^{r\Delta t}-u'}{u'^2-1}-k\left(\frac{ue^{r\Delta t}-1}{u^2-1} - \frac{u'e^{r\Delta t}-1}{u'^2-1} \right).
\end{align*}
The function $p(u)=\frac{ue^{r\Delta t}-1}{u^2-1}$ is a monotonically decreasing function because the derivative $p'(u)=\frac{-e^{r\Delta t}-u^2e^{r\Delta t}+2u}{(u^2-1)^2}$ would have the roots $u_{a,b}=\frac{1\pm \sqrt{1-e^{2r\Delta t}}}{e^{r\Delta t}}$, but due to $e^{r\Delta t}\ge1$ either the derivative has no roots or a root at $u=1$ in the case $r=0$. However, in the specification of the binomial tree it has to hold $u>e^{r\Delta t}$ and therewith, the derivative does not have any roots. From $p'(2)<0$ the monotonically decreasing property follows. Hence, $\left(\frac{ue^{r\Delta t}-1}{u^2-1} - \frac{u'e^{r\Delta t}-1}{u'^2-1} \right)\ge0$.  If additionally the function $g(u)=\frac{u^2e^{r\Delta t}-u}{u^2-1}$ is monotonically increasing, $E[(X-k)^+]\le E[(Y-k)^+]$ follows directly. The derivative of $g$ is given by $g'(u)=\frac{u^2-2ue^{r\Delta t}+1}{(u^2-1)^2}$. The roots of the derivative are $u_{a,b}=e^{r\Delta t}\pm \sqrt{e^{2r\Delta t}-1}$. Thus, for either $u\le e^{r\Delta t}- \sqrt{e^{2r\Delta t}-1}$ or $u\ge e^{r\Delta t}+ \sqrt{e^{2r\Delta t}-1}$ the function is monotonically increasing. From $u\ge e^{r\Delta t}$ in the binomial tree it follows $E[(X-k)^+]\le E[(Y-k)^+]$ if $u\ge e^{r\Delta t}+ \sqrt{e^{2r\Delta t}-1}$.\\

\textit{Case 4:} $\frac{1}{u'}<k<\frac{1}{u}<u<u'$.\\

\noindent In this case one option is in the upward as well as in the downward case in-the-money, whereas the second option is only in-the-money in the corresponding upward case. $E[(X-k)^+]=p(u)(u-k)+(1-p(u))(\frac{1}{u}-k)=e^{r\Delta t} - k$.
\begin{align*}
E[(X-k)^+]-E[(Y-k)^+]&=e^{r\Delta t} - k-p(u')(u'-k)\\
&=e^{r\Delta t} - k - \frac{u'e^{r\Delta t}-1}{u'^2-1}(u'-k)\\
&=\frac{-ku'^2+(1+ke^{r\Delta t})u'-e^{r\Delta t}}{u'^2-1}.
\end{align*}
The roots are given by $u_{a,b}=\frac{(-e^{r\Delta t}k-1)\pm\sqrt{(e^{r\Delta t}k+1)^2-4e^{r\Delta t}k}}{2k}$. Thus by either $u'\le\frac{(-e^{r\Delta t}k-1)-\sqrt{(e^{r\Delta t}k+1)^2-4e^{r\Delta t}k}}{2k}=\frac{1}{k}$ or by $u'\ge\frac{(-e^{r\Delta t}k-1)+\sqrt{(e^{r\Delta t}k+1)^2-4e^{r\Delta t}k}}{2k}=\frac{e^{r\Delta t}}{k}$ the second condition is satisfied.\\

\textit{Case 5:} $k<\frac{1}{u'}<\frac{1}{u}<u<u'$.\\

\noindent Case 4 implies Case 5.

\end{proof}

\begin{lemma}\label{Lemma_Concave}
The function $f:(0,\infty)^n\rightarrow\mathbb{R}$, $f(x)=\prod_{i=1}^n x_i$ is concave for $n$ even.
\end{lemma}
\begin{proof}
The Hessian matrix $H=(h_{ij})_{1\le i,j\le n}$ is given by 
\begin{equation*}
   h_{ij} =
   \begin{cases}
     0 & i=j, \\
     \prod_{k=1,k\neq i, k\neq j}^n x_k & i\neq j.
   \end{cases}
\end{equation*}
Applying the Leibniz formula the determinant of the Hessian matrix is given by
\begin{equation}
\det(H)=\sum_{\sigma\in S_n} sgn(\sigma)\prod_{i=1}^n h_{i\sigma_i},
\end{equation}
where $S_n$ denotes the symmetric group on $n$ elements and $sgn(\sigma)$ is the signature of the permutation $\sigma$, which is 1 if $\sigma$ is even and -1 if $\sigma$ is odd. From $h_{ij}=0$ for $i=j$ it directly follows that we only have to consider derangements, i.e. fixed-point-free permutations. We denote the set of derangements in $S_n$ with $D_n$ and this yields for the determinant of the Hessian matrix,
\begin{equation*}
\det(H)=\sum_{\sigma\in D_n} sgn(\sigma)\prod_{i=1}^n x_i^{n-2}.
\end{equation*}
Next we apply \cite[Theorem 1]{chapman2001involution}, which states that the difference of even and odd derangements in the symmetric group $S_n$ is given by $(-1)^{n-1}(n-1)$. 
Therewith the determinant of the Hessian matrix is given by,
\begin{equation*}
\det(H)=(-1)^{n-1}(n-1)\prod_{i=1}^n x_i^{n-2}.
\end{equation*}
From $x_i\in(0,\infty),\ i=1,\ldots,n$ it follows that the determinant is negative for even $n$ and thus, the function $f$ is concave.
\end{proof}
\section{Detailed Results for Effects of de--Ameri\-cani\-za\-tion on Pricing}
\begin{table}[]\tiny
\centering
\begin{tabular}{lc|cccccccc}

\multicolumn{1}{c}{}                    &   & $T_1$    & $T_2$    & $T_3$   & $T_4$   & $T_5$   & $T_6$   & $T_7$   & $T_8$   \\ \hline
\multicolumn{1}{c}{\multirow{5}{*}{$p_1$}} 	&	 $r=0\%$	&	1.E-4	&	1.E-4	&	1.E-4	&	1.E-4	&	1.E-4	&	1.E-4	&	9.E-5	&	6.E-5	\\
\multicolumn{1}{c}{}                    	&	 $r=1\%$ 	&	3.E-4	&	3.E-4	&	3.E-4	&	3.E-4	&	2.E-4	&	2.E-4	&	2.E-4	&	1.E-4	\\
\multicolumn{1}{c}{}                    	&	 $r=2\%$ 	&	3.E-4	&	3.E-4	&	3.E-4	&	3.E-4	&	2.E-4	&	2.E-4	&	2.E-4	&	8.E-5	\\
\multicolumn{1}{c}{}                    	&	 $r=5\%$ 	&	3.E-4	&	3.E-4	&	3.E-4	&	3.E-4	&	2.E-4	&	2.E-4	&	1.E-4	&	-3.E-5	\\
\multicolumn{1}{c}{}                    	&	 $r=7\%$ 	&	3.E-4	&	3.E-4	&	3.E-4	&	3.E-4	&	2.E-4	&	2.E-4	&	1.E-4	&	-9.E-5	\\
\hline																			
\multicolumn{1}{c}{\multirow{5}{*}{$p_2$}} 	&	 $r=0\%$	&	2.E-4	&	2.E-4	&	2.E-4	&	1.E-4	&	1.E-4	&	1.E-4	&	1.E-4	&	7.E-5	\\
\multicolumn{1}{c}{}                    	&	 $r=1\%$ 	&	3.E-4	&	3.E-4	&	3.E-4	&	3.E-4	&	3.E-4	&	2.E-4	&	2.E-4	&	1.E-4	\\
\multicolumn{1}{c}{}                    	&	 $r=2\%$ 	&	3.E-4	&	3.E-4	&	3.E-4	&	3.E-4	&	3.E-4	&	2.E-4	&	2.E-4	&	5.E-5	\\
\multicolumn{1}{c}{}                    	&	 $r=5\%$ 	&	3.E-4	&	4.E-4	&	3.E-4	&	3.E-4	&	2.E-4	&	1.E-4	&	6.E-5	&	3.E-4	\\
\multicolumn{1}{c}{}                    	&	 $r=7\%$ 	&	3.E-4	&	3.E-4	&	3.E-4	&	3.E-4	&	2.E-4	&	7.E-5	&	7.E-5	&	5.E-4	\\
\hline																			
\multicolumn{1}{c}{\multirow{5}{*}{$p_3$}} 	&	 $r=0\%$	&	2.E-4	&	2.E-4	&	2.E-4	&	1.E-4	&	1.E-4	&	1.E-4	&	1.E-4	&	6.E-5	\\
\multicolumn{1}{c}{}                    	&	 $r=1\%$ 	&	3.E-4	&	3.E-4	&	3.E-4	&	3.E-4	&	3.E-4	&	2.E-4	&	2.E-4	&	1.E-4	\\
\multicolumn{1}{c}{}                    	&	 $r=2\%$ 	&	3.E-4	&	3.E-4	&	3.E-4	&	3.E-4	&	3.E-4	&	2.E-4	&	2.E-4	&	6.E-5	\\
\multicolumn{1}{c}{}                    	&	 $r=5\%$ 	&	3.E-4	&	3.E-4	&	3.E-4	&	3.E-4	&	2.E-4	&	1.E-4	&	8.E-5	&	-2.E-4	\\
\multicolumn{1}{c}{}                    	&	 $r=7\%$ 	&	3.E-4	&	3.E-4	&	3.E-4	&	3.E-4	&	2.E-4	&	1.E-4	&	-1.E-7	&	-3.E-4	\\
\hline																			
\multicolumn{1}{c}{\multirow{5}{*}{$p_4$}} 	&	 $r=0\%$	&	2.E-4	&	1.E-4	&	1.E-4	&	1.E-4	&	1.E-4	&	8.E-5	&	8.E-5	&	2.E-4	\\
\multicolumn{1}{c}{}                    	&	 $r=1\%$ 	&	3.E-4	&	3.E-4	&	3.E-4	&	2.E-4	&	2.E-4	&	2.E-4	&	2.E-4	&	3.E-4	\\
\multicolumn{1}{c}{}                    	&	 $r=2\%$ 	&	3.E-4	&	3.E-4	&	3.E-4	&	2.E-4	&	2.E-4	&	1.E-4	&	1.E-4	&	3.E-4	\\
\multicolumn{1}{c}{}                    	&	 $r=5\%$ 	&	3.E-4	&	3.E-4	&	3.E-4	&	2.E-4	&	2.E-4	&	1.E-4	&	8.E-5	&	1.E-4	\\
\multicolumn{1}{c}{}                    	&	 $r=7\%$ 	&	4.E-4	&	3.E-4	&	3.E-4	&	2.E-4	&	2.E-4	&	1.E-4	&	7.E-5	&	2.E-4	\\
\hline																			
\multicolumn{1}{c}{\multirow{5}{*}{$p_5$}} 	&	 $r=0\%$	&	2.E-4	&	2.E-4	&	2.E-4	&	1.E-4	&	1.E-4	&	1.E-4	&	9.E-5	&	4.E-5	\\
\multicolumn{1}{c}{}                    	&	 $r=1\%$ 	&	3.E-4	&	3.E-4	&	3.E-4	&	3.E-4	&	3.E-4	&	2.E-4	&	2.E-4	&	6.E-7	\\
\multicolumn{1}{c}{}                    	&	 $r=2\%$ 	&	3.E-4	&	3.E-4	&	3.E-4	&	3.E-4	&	2.E-4	&	2.E-4	&	1.E-4	&	-1.E-4	\\
\multicolumn{1}{c}{}                    	&	 $r=5\%$ 	&	4.E-4	&	3.E-4	&	3.E-4	&	3.E-4	&	2.E-4	&	1.E-4	&	-7.E-6	&	-5.E-4	\\
\multicolumn{1}{c}{}                    	&	 $r=7\%$ 	&	4.E-4	&	4.E-4	&	3.E-4	&	3.E-4	&	2.E-4	&	4.E-5	&	-1.E-4	&	-8.E-4	\\
\hline
\end{tabular}
\caption{De--Ameri\-cani\-za\-tion effects on pricing put options in the CEV model - Average error between the de-Americanized and European prices for each maturity.}
\label{Tab_Pricing_CEV_avg_errors}
\end{table}

\begin{table}[]\tiny
\centering
\begin{tabular}{lc|ccccccccc}
\multicolumn{1}{c}{}                    &   & $0.80$    & $0.85$    & $0.90$   & $0.95$   & $1.00$   & $1.05$   & $1.10$   & $1.15$ & $1.20$   \\ \hline
\multicolumn{1}{c}{\multirow{5}{*}{$p_1$}} 	&	 $r=0\%$	&	4.E-5	&	7.E-5	&	1.E-4	&	2.E-4	&	2.E-4	&	2.E-4	&	1.E-4	&	1.E-4	&	8.E-5	\\
\multicolumn{1}{c}{}                    	&	 $r=1\%$ 	&	8.E-5	&	1.E-4	&	2.E-4	&	3.E-4	&	4.E-4	&	4.E-4	&	3.E-4	&	2.E-4	&	2.E-4	\\
\multicolumn{1}{c}{}                    	&	 $r=2\%$ 	&	8.E-5	&	1.E-4	&	2.E-4	&	3.E-4	&	4.E-4	&	4.E-4	&	3.E-4	&	2.E-4	&	2.E-4	\\
\multicolumn{1}{c}{}                    	&	 $r=5\%$ 	&	7.E-5	&	1.E-4	&	2.E-4	&	3.E-4	&	4.E-4	&	4.E-4	&	2.E-4	&	2.E-4	&	1.E-4	\\
\multicolumn{1}{c}{}                    	&	 $r=7\%$ 	&	7.E-5	&	1.E-4	&	2.E-4	&	3.E-4	&	4.E-4	&	4.E-4	&	1.E-4	&	2.E-4	&	8.E-5	\\
\hline																					
\multicolumn{1}{c}{\multirow{5}{*}{$p_2$}} 	&	 $r=0\%$	&	5.E-5	&	9.E-5	&	1.E-4	&	2.E-4	&	2.E-4	&	2.E-4	&	1.E-4	&	1.E-4	&	8.E-5	\\
\multicolumn{1}{c}{}                    	&	 $r=1\%$ 	&	9.E-5	&	2.E-4	&	3.E-4	&	3.E-4	&	4.E-4	&	4.E-4	&	3.E-4	&	3.E-4	&	2.E-4	\\
\multicolumn{1}{c}{}                    	&	 $r=2\%$ 	&	8.E-5	&	2.E-4	&	2.E-4	&	3.E-4	&	4.E-4	&	4.E-4	&	3.E-4	&	2.E-4	&	2.E-4	\\
\multicolumn{1}{c}{}                    	&	 $r=5\%$ 	&	6.E-5	&	1.E-4	&	2.E-4	&	3.E-4	&	4.E-4	&	4.E-4	&	3.E-4	&	3.E-4	&	2.E-4	\\
\multicolumn{1}{c}{}                    	&	 $r=7\%$ 	&	6.E-5	&	1.E-4	&	2.E-4	&	3.E-4	&	4.E-4	&	4.E-4	&	3.E-4	&	3.E-4	&	2.E-4	\\
\hline																					
\multicolumn{1}{c}{\multirow{5}{*}{$p_3$}} 	&	 $r=0\%$	&	4.E-5	&	7.E-5	&	1.E-4	&	2.E-4	&	2.E-4	&	2.E-4	&	1.E-4	&	1.E-4	&	8.E-5	\\
\multicolumn{1}{c}{}                    	&	 $r=1\%$ 	&	7.E-5	&	1.E-4	&	2.E-4	&	3.E-4	&	4.E-4	&	4.E-4	&	3.E-4	&	3.E-4	&	2.E-4	\\
\multicolumn{1}{c}{}                    	&	 $r=2\%$ 	&	7.E-5	&	1.E-4	&	2.E-4	&	3.E-4	&	4.E-4	&	4.E-4	&	3.E-4	&	3.E-4	&	2.E-4	\\
\multicolumn{1}{c}{}                    	&	 $r=5\%$ 	&	5.E-5	&	1.E-4	&	2.E-4	&	3.E-4	&	4.E-4	&	3.E-4	&	2.E-4	&	2.E-4	&	6.E-5	\\
\multicolumn{1}{c}{}                    	&	 $r=7\%$ 	&	3.E-5	&	9.E-5	&	2.E-4	&	2.E-4	&	3.E-4	&	3.E-4	&	2.E-4	&	1.E-4	&	-2.E-5	\\
\hline																					
\multicolumn{1}{c}{\multirow{5}{*}{$p_4$}} 	&	 $r=0\%$	&	4.E-5	&	7.E-5	&	1.E-4	&	1.E-4	&	2.E-4	&	2.E-4	&	1.E-4	&	1.E-4	&	9.E-5	\\
\multicolumn{1}{c}{}                    	&	 $r=1\%$ 	&	9.E-5	&	1.E-4	&	2.E-4	&	3.E-4	&	4.E-4	&	4.E-4	&	2.E-4	&	3.E-4	&	2.E-4	\\
\multicolumn{1}{c}{}                    	&	 $r=2\%$ 	&	8.E-5	&	1.E-4	&	2.E-4	&	3.E-4	&	4.E-4	&	4.E-4	&	2.E-4	&	2.E-4	&	1.E-4	\\
\multicolumn{1}{c}{}                    	&	 $r=5\%$ 	&	9.E-5	&	1.E-4	&	2.E-4	&	3.E-4	&	4.E-4	&	4.E-4	&	2.E-4	&	2.E-4	&	-3.E-5	\\
\multicolumn{1}{c}{}                    	&	 $r=7\%$ 	&	9.E-5	&	1.E-4	&	2.E-4	&	3.E-4	&	4.E-4	&	4.E-4	&	2.E-4	&	2.E-4	&	-1.E-4	\\
\hline																					
\multicolumn{1}{c}{\multirow{5}{*}{$p_5$}} 	&	 $r=0\%$	&	3.E-5	&	7.E-5	&	1.E-4	&	2.E-4	&	2.E-4	&	2.E-4	&	1.E-4	&	1.E-4	&	7.E-5	\\
\multicolumn{1}{c}{}                    	&	 $r=1\%$ 	&	6.E-5	&	1.E-4	&	2.E-4	&	3.E-4	&	4.E-4	&	4.E-4	&	2.E-4	&	2.E-4	&	1.E-4	\\
\multicolumn{1}{c}{}                    	&	 $r=2\%$ 	&	5.E-5	&	1.E-4	&	2.E-4	&	3.E-4	&	4.E-4	&	4.E-4	&	2.E-4	&	2.E-4	&	1.E-4	\\
\multicolumn{1}{c}{}                    	&	 $r=5\%$ 	&	2.E-5	&	7.E-5	&	1.E-4	&	2.E-4	&	3.E-4	&	3.E-4	&	1.E-4	&	1.E-4	&	-9.E-6	\\
\multicolumn{1}{c}{}                    	&	 $r=7\%$ 	&	4.E-5	&	4.E-5	&	1.E-4	&	2.E-4	&	2.E-4	&	2.E-4	&	4.E-5	&	4.E-5	&	-1.E-4	\\
\hline
\end{tabular}
\caption{De--Ameri\-cani\-za\-tion effects on pricing put options in the CEV model - Average error between the de-Americanized and European prices for each strike.}
\label{Tab_Pricing_CEV_avg_errors_strikes}
\end{table}

\begin{table}[]\tiny
\centering
\begin{tabular}{lc|cccccccc}
                                        &   & $T_1$    & $T_2$    & $T_3$   & $T_4$   & $T_5$   & $T_6$   & $T_7$   & $T_8$   \\ \hline
\multirow{5}{*}{$p_1$}                     & $r=0\%$ &   0.055&	0.104&	0.108&	0.154&	0.204&	0.211&	0.218&	0.244    \\
                                        & $r=1\%$ &      0.055&	0.102&	0.106&	0.151&	0.155&	0.203&	0.208&	0.226    \\
                                        & $r=2\%$ &      0.054&	0.101&	0.104&	0.107&	0.151&	0.196&	0.198&	0.208     \\
                                        & $r=5\%$ &     0.052&	0.057&	0.097&	0.098&	0.100&	0.136&	0.136&	0.161    \\
                                        & $r=7\%$ &     0.051&	0.054&	0.057&	0.093&	0.093&	0.093&	0.093&	0.109   \\
 \hline
\multirow{5}{*}{$p_2$}                     & $r=0\%$ &    0.103&	0.154&	0.204&	0.208&	0.217&	0.230&	0.242&	0.283     \\
                                        & $r=1\%$ &     0.103&	0.152&	0.201&	0.204&	0.212&	0.223&	0.233&	0.265     \\
                                        & $r=2\%$ &    0.102&	0.109&	0.154&	0.201&	0.207&	0.216&	0.224&	0.248     \\
                                        & $r=5\%$ &     0.099&	0.104&	0.147&	0.150&	0.192&	0.195&	0.198&	0.203  \\
                                        & $r=7\%$ &     0.058&	0.101&	0.105&	0.144&	0.146&	0.182&	0.181&	0.176  \\
 \hline
\multirow{5}{*}{$p_3$}                     & $r=0\%$ &    0.152&	0.205&	0.212&	0.219&	0.233&	0.252&	0.269&	0.320    \\
                                        & $r=1\%$ &      0.108&	0.203&	0.209&	0.216&	0.228&	0.245&	0.260&	0.303    \\
                                        & $r=2\%$ &     0.107&	0.201&	0.207&	0.213&	0.224&	0.239&	0.251&	0.287    \\
                                        & $r=5\%$ &     0.105&	0.153&	0.199&	0.203&	0.210&	0.219&	0.226&	0.241    \\
                                        & $r=7\%$ &      0.103&	0.150&	0.194&	0.196&	0.201&	0.206&	0.210&	0.214\\   
  \hline
\multirow{5}{*}{$p_4$}                     & $r=0\%$ &      0.156&	0.212&	0.223&	0.233&	0.252&	0.276&	0.297&	0.352    \\
                                        & $r=1\%$ &      0.155&	0.210&	0.220&	0.230&	0.247&	0.270&	0.288	&0.336    \\
                                        & $r=2\%$ &      0.154&	0.208&	0.218&	0.227&	0.243&	0.263&	0.279&	0.319    \\
                                        & $r=5\%$ &      0.152&	0.203&	0.210&	0.217&	0.229&	0.244&	0.255&	0.275    \\
                                        & $r=7\%$ &      0.150&	0.200&	0.206&	0.211&	0.221&	0.232&	0.239&	0.247\\                                          
                                                                                  \hline
\multirow{5}{*}{$p_5$}                     & $r=0\%$ &     0.205&	0.220&	0.235&	0.248&	0.272&	0.300&	0.323&	0.377    \\
                                        & $r=1\%$ &     0.205&	0.219&	0.232&	0.245&	0.267&	0.294&	0.314&	0.360     \\
                                        & $r=2\%$ &     0.204&	0.217&	0.230&	0.242&	0.263&	0.287&	0.306&	0.345    \\
                                        & $r=5\%$ &     0.201&	0.212&	0.223&	0.233&	0.250&	0.268&	0.281&	0.300  \\
                                        & $r=7\%$ &     0.156&	0.209&	0.218&	0.227&	0.241&	0.256&	0.266&	0.273\\
\hline										
\end{tabular}
\caption{De--Ameri\-cani\-za\-tion effects on pricing put options in the CEV model - Maximal European put prices.}
\label{Tab_Pricing_CEV_max_prices}
\end{table}
\begin{table}[]\tiny
\centering
\begin{tabular}{lc|cccccccc}
\multicolumn{1}{c}{}                    &   & $T_1$    & $T_2$    & $T_3$   & $T_4$   & $T_5$   & $T_6$   & $T_7$   & $T_8$   \\ \hline
\multicolumn{1}{c}{\multirow{5}{*}{$p_1$}} & $r=0\%$ & 1.E-8   & 2.E-7  & 2.E-7  & 1.E-7  & 2.E-7 &   1.E-7 & 3.E-7 & -3.E-7     \\
\multicolumn{1}{c}{}                    & $r=1\%$ &    -7.E-5  & -4.E-5 & -5.E-5 & -4.E-5 & -4.E-5 & -3.E-5 & -3.E-5 & 6.E-6    \\
\multicolumn{1}{c}{}                    & $r=2\%$ &     -9.E-5 & -7.E-5 & -9.E-5 & -6.E-5 & -5.E-5 & -5.E-5 & -4.E-5 & 3.E-5 	   \\
\multicolumn{1}{c}{}                    & $r=5\%$ &     -3.E-4 & -4.E-4 & -1.E-4 & -9.E-5 & -9.E-5 & -8.E-5 & -7.E-5 & 3.E-5    \\
\multicolumn{1}{c}{}                    & $r=7\%$ &   -3.E-4   & -3.E-4 & -9.E-4 & -1.E-4 & -8.E-5 & -1.E-4 & -1.E-4 & -4.E-5     \\
\hline

\multicolumn{1}{c}{\multirow{5}{*}{$p_2$}} & $r=0\%$ &   3.E-8  & 1.E-7   & 1.E-8 & -2.E-7 & 3.E-7  & 4.E-7 & -8.E-8 & -3.E-8     \\
\multicolumn{1}{c}{}                    & $r=1\%$ &      -4.E-5 & -3.E-5 & -4.E-5 & -3.E-5 & -2.E-5 & -1.E-5 & -2.E-6 & 9.E-5	   \\
\multicolumn{1}{c}{}                    & $r=2\%$ &      -1.E-4 & -5.E-5 & -6.E-5 & -5.E-5 & -4.E-5 & -3.E-5 & -8.E-6 & 1.E-4   \\
\multicolumn{1}{c}{}                    & $r=5\%$ &      -2.E-4 & -7.E-5 & -1.E-4 & -9.E-5 & -8.E-5 & -9.E-5 & -1.E-4 & 2.E-5     \\
\multicolumn{1}{c}{}                    & $r=7\%$ &      -2.E-4 & -6.E-4 & -1.E-4 & -1.E-4 & -1.E-4 & -2.E-4 & -2.E-4 & -3.E-4   \\
\hline

\multicolumn{1}{c}{\multirow{5}{*}{$p_3$}} & $r=0\%$ &     3.E-9 & 4.E-8 & -3.E-7   & 3.E-9 & 8.E-8 & -2.E-7 & 5.E-7  & 1.E-7    \\
\multicolumn{1}{c}{}                    & $r=1\%$ &       -3.E-5 & -2.E-5 & -3.E-5 & -2.E-5 & -6.E-6 & 1.E-5& 3.E-5   & 2.E-4  \\
\multicolumn{1}{c}{}                    & $r=2\%$ &       -6.E-5 & -4.E-5 & -4.E-5 & -3.E-5 & -7.E-6 & 2.E-5 & 6.E-5  & 3.E-4    \\
\multicolumn{1}{c}{}                    & $r=5\%$ &       -1.E-4 & -7.E-5 & -9.E-5 & -5.E-5 & -2.E-5 & 2.E-5 & 7.E-5  & 5.E-4     \\
\multicolumn{1}{c}{}                    & $r=7\%$ &       -4.E-4 & -9.E-5 & -1.E-4 & -7.E-5 & -4.E-5 & -2.E-5 & 1.E-5 & 4.E-4   \\
\hline

\multicolumn{1}{c}{\multirow{5}{*}{$p_4$}} & $r=0\%$ &   2.E-8&	-8.E-8&	2.E-7&	4.E-8 &	-5.E-7&	-1.E-7&	-4.E-7&	2.E-7    \\
\multicolumn{1}{c}{}                    & $r=1\%$ &    -3.E-5&	-2.E-5&	-2.E-5&	-1.E-5&	2.E-6&	2.E-5&	5.E-5&	2.E-4	    \\
\multicolumn{1}{c}{}                    & $r=2\%$ &    -5.E-5&	-3.E-5&	-3.E-5&	-1.E-5&	1.E-5&	5.E-5&	9.E-5&	4.E-4    \\
\multicolumn{1}{c}{}                    & $r=5\%$ &    -1.E-4&	-5.E-5&	-6.E-5&	-1.E-5&	4.E-5&	1.E-4&	2.E-4&	8.E-4   \\
\multicolumn{1}{c}{}                    & $r=7\%$ &    -1.E-4&	-6.E-5&	-8.E-5&	-1.E-5&	4.E-5&	1.E-4&	2.E-4&	9.E-4     \\
\hline

\multicolumn{1}{c}{\multirow{5}{*}{$p_5$}} & $r=0\%$ &  -1.E-7&	-3.E-7&	-3.E-7&	2.E-7&	5.E-8&	8.E-7&	1.E-6&	-2.E-6    \\
\multicolumn{1}{c}{}                    & $r=1\%$ &     -2.E-5&	-2.E-5&	-2.E-5&	-7.E-6&	2.E-5&	8.E-5&	2.E-4&	7.E-4    \\
\multicolumn{1}{c}{}                    & $r=2\%$ &     -4.E-5&	-3.E-5&	-4.E-5&	-1.E-5&	3.E-5&	1.E-4&	3.E-4&	1.E-3   \\
\multicolumn{1}{c}{}                    & $r=5\%$ &     -1.E-4&	-7.E-5&	-1.E-4&	-5.E-5&	2.E-5&	2.E-4&	5.E-4&	2.E-3  \\
\multicolumn{1}{c}{}                    & $r=7\%$ &     -1.E-4&	-1.E-4&	-1.E-4&	-9.E-5&	-1.E-5&	2.E-4&	5.E-4&	3.E-3    \\
\end{tabular}
\caption{De--Ameri\-cani\-za\-tion effects on pricing put options in the Heston model - Average error between the de-Americanized and European prices for each maturity.}
\label{Tab_Pricing_Heston_avg_errors}
\end{table}

\begin{table}[]\tiny
\centering
\begin{tabular}{lc|ccccccccc}
\multicolumn{1}{c}{}                    &   & $0.80$    & $0.85$    & $0.90$   & $0.95$   & $1.00$   & $1.05$   & $1.10$   & $1.15$ & $1.20$   \\ \hline
\multicolumn{1}{c}{\multirow{5}{*}{$p_1$}} & $r=0\%$ &	-4.E-9	&	5.E-8	&	1.E-8	&	-1.E-7	&	-6.E-8	&	-5.E-7	&	-3.E-8	&	-2.E-7	&	8.E-7	\\
\multicolumn{1}{c}{}                    & $r=1\%$ &	3.E-6	&	6.E-6	&	9.E-6	&	1.E-5	&	2.E-5	&	4.E-5	&	7.E-5	&	2.E-4	&	3.E-4	\\
\multicolumn{1}{c}{}                    & $r=2\%$ &	8.E-6	&	1.E-5	&	2.E-5	&	3.E-5	&	5.E-5	&	9.E-5	&	2.E-4	&	5.E-4	&	9.E-4	\\
\multicolumn{1}{c}{}                    & $r=5\%$ &	2.E-5	&	3.E-5	&	5.E-5	&	8.E-5	&	1.E-4	&	2.E-4	&	6.E-4	&	2.E-3	&	4.E-3	\\
\multicolumn{1}{c}{}                    & $r=7\%$ &	2.E-5	&	4.E-5	&	7.E-5	&	1.E-4	&	2.E-4	&	4.E-4	&	9.E-4	&	3.E-3	&		\\
\hline
\multicolumn{1}{c}{\multirow{5}{*}{$p_2$}} & $r=0\%$ &	3.E-8	&	-6.E-8	&	8.E-9	&	1.E-7	&	-2.E-7	&	-4.E-7	&	-3.E-7	&	5.E-7	&	-6.E-7	\\
\multicolumn{1}{c}{}                    & $r=1\%$ &	-2.E-5	&	-1.E-5	&	-2.E-6	&	7.E-6	&	2.E-5	&	3.E-5	&	4.E-5	&	1.E-4	&	2.E-4	\\
\multicolumn{1}{c}{}                    & $r=2\%$ &	-2.E-5	&	-2.E-6	&	1.E-5	&	3.E-5	&	4.E-5	&	6.E-5	&	1.E-4	&	2.E-4	&	4.E-4	\\
\multicolumn{1}{c}{}                    & $r=5\%$ &	1.E-6	&	3.E-5	&	6.E-5	&	9.E-5	&	1.E-4	&	2.E-4	&	3.E-4	&	8.E-4	&	1.E-3	\\
\multicolumn{1}{c}{}                    & $r=7\%$ &	1.E-5	&	5.E-5	&	8.E-5	&	1.E-4	&	2.E-4	&	3.E-4	&	4.E-4	&	1.E-3	&	2.E-3	\\
\hline
\multicolumn{1}{c}{\multirow{5}{*}{$p_3$}} & $r=0\%$ &	7.E-8	&	8.E-7	&	7.E-7	&	8.E-7	&	7.E-7	&	2.E-6	&	7.E-7	&	3.E-6	&	2.E-6	\\
\multicolumn{1}{c}{}                    & $r=1\%$ &	2.E-6	&	4.E-6	&	9.E-6	&	1.E-5	&	2.E-5	&	3.E-5	&	4.E-5	&	6.E-5	&	1.E-4	\\
\multicolumn{1}{c}{}                    & $r=2\%$ &	5.E-6	&	1.E-5	&	2.E-5	&	4.E-5	&	6.E-5	&	8.E-5	&	9.E-5	&	2.E-4	&	3.E-4	\\
\multicolumn{1}{c}{}                    & $r=5\%$ &	2.E-5	&	5.E-5	&	9.E-5	&	1.E-4	&	2.E-4	&	2.E-4	&	4.E-4	&	8.E-4	&	1.E-3	\\
\multicolumn{1}{c}{}                    & $r=7\%$ &	4.E-5	&	9.E-5	&	1.E-4	&	2.E-4	&	3.E-4	&	3.E-4	&	7.E-4	&	1.E-3	&	2.E-3	\\
\hline
\multicolumn{1}{c}{\multirow{5}{*}{$p_4$}} & $r=0\%$ &	-2.E-8	&	-8.E-8	&	2.E-7	&	-1.E-8	&	1.E-7	&	-7.E-8	&	1.E-7	&	4.E-9	&	-2.E-7	\\
\multicolumn{1}{c}{}                    & $r=1\%$ &	-2.E-4	&	-2.E-4	&	-2.E-4	&	-2.E-4	&	-2.E-4	&	-2.E-4	&	-1.E-4	&	-9.E-5	&	-7.E-5	\\
\multicolumn{1}{c}{}                    & $r=2\%$ &	-4.E-4	&	-4.E-4	&	-3.E-4	&	-3.E-4	&	-2.E-4	&	-2.E-4	&	-1.E-4	&	-7.E-5	&	-9.E-6	\\
\multicolumn{1}{c}{}                    & $r=5\%$ &	-7.E-4	&	-6.E-4	&	-5.E-4	&	-4.E-4	&	-2.E-4	&	-1.E-4	&	-5.E-5	&	1.E-4	&	3.E-4	\\
\multicolumn{1}{c}{}                    & $r=7\%$ &	-7.E-4	&	-6.E-4	&	-5.E-4	&	-3.E-4	&	-2.E-4	&	-1.E-4	&	3.E-5	&	3.E-4	&	6.E-4	\\
\hline
\multicolumn{1}{c}{\multirow{5}{*}{$p_5$}} & $r=0\%$ &	2.E-7	&	2.E-7	&	4.E-7	&	7.E-7	&	7.E-7	&	4.E-7	&	8.E-7	&	1.E-6	&	9.E-7	\\
\multicolumn{1}{c}{}                    & $r=1\%$ &	4.E-4	&	4.E-4	&	4.E-4	&	4.E-4	&	3.E-4	&	3.E-4	&	3.E-4	&	3.E-4	&	3.E-4	\\
\multicolumn{1}{c}{}                    & $r=2\%$ &	7.E-4	&	7.E-4	&	6.E-4	&	6.E-4	&	6.E-4	&	5.E-4	&	5.E-4	&	4.E-4	&	4.E-4	\\
\multicolumn{1}{c}{}                    & $r=5\%$ &	1.E-3	&	1.E-3	&	1.E-3	&	1.E-3	&	9.E-4	&	7.E-4	&	6.E-4	&	5.E-4	&	5.E-4	\\
\multicolumn{1}{c}{}                    & $r=7\%$ &	3.E-4	&	1.E-3	&	1.E-3	&	1.E-3	&	9.E-4	&	7.E-4	&	6.E-4	&	6.E-4	&	5.E-4	\\
\hline
\end{tabular}
\caption{De--Ameri\-cani\-za\-tion effects on pricing put options in the Heston model - Average error between the de-Americanized and European prices for each strike. The empty fields are due to Remark \ref{Remark_Exercise_Region}.}
\label{Tab_Pricing_Heston_avg_errors_strikes}
\end{table}

\begin{table}[]\tiny
\centering
\begin{tabular}{lc|cccccccc}

                                        &   & $T_1$    & $T_2$    & $T_3$   & $T_4$   & $T_5$   & $T_6$   & $T_7$   & $T_8$   \\ \hline
\multirow{5}{*}{$p_1$}                     & $r=0\%$ &    0.200	&	0.202	&	0.205	&	0.209	&	0.217	&	0.229	&	0.240	&	0.278    \\
                                        & $r=1\%$ &      0.199	&	0.200	&	0.202	&	0.205	&	0.212	&	0.222	&	0.231	&	0.261     \\
                                        & $r=2\%$ &      0.198	&	0.198	&	0.200	&	0.202	&	0.207	&	0.214	&	0.221	&	0.244     \\
                                        & $r=5\%$ &     0.195	&	0.193	&	0.191	&	0.191	&	0.192	&	0.194	&	0.195	&	0.198    \\
                                        & $r=7\%$ &     0.194	&	0.189	&	0.186	&	0.184	&	0.182	&	0.180	&	0.179	&	0.171   \\
 \hline
\multirow{5}{*}{$p_2$}                     & $r=0\%$ &    0.201	&	0.204	&	0.208	&	0.214	&	0.224	&	0.238	&	0.251	&	0.293     \\
                                        & $r=1\%$ &     0.200	&	0.202	&	0.206	&	0.210	&	0.219	&	0.231	&	0.242	&	0.277     \\
                                        & $r=2\%$ &     0.199	&	0.200	&	0.203	&	0.207	&	0.214	&	0.224	&	0.233	&	0.261     \\
                                        & $r=5\%$ &     0.196	&	0.195	&	0.195	&	0.197	&	0.200	&	0.204	&	0.208	&	0.217  \\
                                        & $r=7\%$ &     0.194	&	0.191	&	0.190	&	0.190	&	0.191	&	0.192	&	0.193	&	0.192  \\
 \hline
\multirow{5}{*}{$p_3$}                     & $r=0\%$ &     0.202	&	0.208	&	0.216	&	0.224	&	0.238	&	0.256	&	0.273	&	0.326    \\
                                        & $r=1\%$ &      0.201	&	0.207	&	0.213	&	0.220	&	0.233	&	0.250	&	0.264	&	0.309    \\
                                        & $r=2\%$ &     0.200	&	0.205	&	0.211	&	0.217	&	0.228	&	0.243	&	0.255	&	0.294    \\
                                        & $r=5\%$ &     0.197	&	0.199	&	0.203	&	0.207	&	0.215	&	0.224	&	0.231	&	0.250    \\
                                        & $r=7\%$ &      0.195	&	0.196	&	0.198	&	0.201	&	0.206	&	0.212	&	0.216	&	0.225\\   
  \hline
\multirow{5}{*}{$p_4$}                     & $r=0\%$ &      0.204	&	0.214	&	0.224	&	0.234	&	0.252	&	0.275	&	0.295	&	0.362    \\
                                        & $r=1\%$ &      0.203	&	0.212	&	0.221	&	0.231	&	0.247	&	0.268	&	0.287	&	0.345    \\
                                        & $r=2\%$ &      0.202	&	0.210	&	0.219	&	0.227	&	0.243	&	0.262	&	0.278	&	0.329    \\
                                        & $r=5\%$ &      0.199	&	0.205	&	0.211	&	0.218	&	0.229	&	0.243	&	0.255	&	0.286    \\
                                        & $r=7\%$ &      0.197	&	0.202	&	0.207	&	0.212	&	0.221	&	0.231	&	0.240	&	0.259\\                                          
                                                                                  \hline
\multirow{5}{*}{$p_5$}                     & $r=0\%$ &     0.207	&	0.222	&	0.234	&	0.248	&	0.270	&	0.297	&	0.322	&	0.400    \\
                                        & $r=1\%$ &     0.206	&	0.220	&	0.232	&	0.245	&	0.265	&	0.291	&	0.314	&	0.384     \\
                                        & $r=2\%$ &     0.205	&	0.218	&	0.230	&	0.242	&	0.261	&	0.285	&	0.306	&	0.369    \\
                                        & $r=5\%$ &     0.202	&	0.213	&	0.223	&	0.233	&	0.249	&	0.268	&	0.283	&	0.326  \\
                                        & $r=7\%$ &     0.201	&	0.210	&	0.219	&	0.227	&	0.241	&	0.257	&	0.269	&	0.300\\
\hline										
\end{tabular}
\caption{De--Ameri\-cani\-za\-tion effects on pricing put options in the Heston model - Maximal European put prices.}
\label{Tab_Pricing_Heston_max_prices}
\end{table}
\begin{table}[]\tiny
\centering
\begin{tabular}{lc|cccccccc}
\multicolumn{1}{c}{}                    &   & $T_1$    & $T_2$    & $T_3$   & $T_4$   & $T_5$   & $T_6$   & $T_7$   & $T_8$   \\ \hline
\multicolumn{1}{c}{\multirow{5}{*}{$p_1$}} & $r=0\%$ & -3.E-4 & -3.E-4  & -2.E-4 & -2.E-4 & -2.E-4 & -2.E-4 & -2.E-4 & -1.E-4     \\
\multicolumn{1}{c}{}                    & $r=1\%$ &    -3.E-4 & -3.E-4  & -3.E-4 & -3.E-4 & -2.E-4 & -2.E-4 & -2.E-4 & -2.E-4    \\
\multicolumn{1}{c}{}                    & $r=2\%$ &     -3.E-4 & -3.E-4 & -3.E-4 & -2.E-4 & -3.E-4 & -2.E-4 & -2.E-4 & -2.E-4 	   \\
\multicolumn{1}{c}{}                    & $r=5\%$ &     -2.E-4 & -2.E-4 & -1.E-4 & -2.E-4 & -2.E-4 & -3.E-5 & -9.E-5 & 9.E-5    \\
\multicolumn{1}{c}{}                    & $r=7\%$ &   -2.E-4 & -2.E-4   & -4.E-5 & -9.E-6 & -5.E-5 & -5.E-5 & 3.E-4 & 9.E-5     \\
\hline

\multicolumn{1}{c}{\multirow{5}{*}{$p_2$}} & $r=0\%$ &   -1.E-4 & -1.E-4 & -1.E-4 & -1.E-4 & -1.E-4 & -1.E-4 & -1.E-4 & -8.E-5     \\
\multicolumn{1}{c}{}                    & $r=1\%$ &      -1.E-4 & -2.E-4& -2.E-4 & -2.E-4 & -2.E-4  & -2.E-4 & -3.E-4 & -3.E-4	   \\
\multicolumn{1}{c}{}                    & $r=2\%$ &      -1.E-4 & -1.E-4 & -1.E-4 & -2.E-4 & -2.E-4 & -3.E-4 & -3.E-4 & -5.E-4   \\
\multicolumn{1}{c}{}                    & $r=5\%$ &      -2.E-5 & -3.E-5 & -1.E-4 & -3.E-5 & -6.E-5 & -3.E-4 & -2.E-4 & -5.E-4     \\
\multicolumn{1}{c}{}                    & $r=7\%$ &      -4.E-5 & 9.E-5  & -5.E-6 & -1.E-4 & 3.E-5  & 6.E-5  & -2.E-4 & -4.E-4   \\
\hline

\multicolumn{1}{c}{\multirow{5}{*}{$p_3$}} & $r=0\%$ &  -1.E-4 & -9.E-5 & -8.E-5 & -7.E-5 & -6.E-5 & -6.E-5 & -5.E-5 & -4.E-5    \\
\multicolumn{1}{c}{}                    & $r=1\%$ &     -1.E-4 & -1.E-4 & -2.E-4 & -2.E-4 & -3.E-4 & -3.E-4 & -4.E-4 & -6.E-4  \\
\multicolumn{1}{c}{}                    & $r=2\%$ &     -1.E-4 & -2.E-4 & -2.E-4 & -3.E-4  &-4.E-4 & -6.E-4 & -7.E-4 & -1.E-3    \\
\multicolumn{1}{c}{}                    & $r=5\%$ &     -1.E-4 & -2.E-4 & -3.E-4 & -4.E-4 & -7.E-4 & -1.E-3 & -2.E-3 & -2.E-3     \\
\multicolumn{1}{c}{}                    & $r=7\%$ &     -9.E-5 & -2.E-4 & -3.E-4 & -5.E-4 & -8.E-4 & -1.E-3 & -2.E-3 & -3.E-3   \\
\hline

\multicolumn{1}{c}{\multirow{5}{*}{$p_4$}} & $r=0\%$ & 1.E-6&  -2.E-6&	-3.E-7&	4.E-7&	1.E-7&	-3.E-7&	3.E-6&	-1.E-6    \\
\multicolumn{1}{c}{}                    & $r=1\%$ &    8.E-5&	1.E-4&	1.E-4&	8.E-5&	5.E-5&	-8.E-6&	-8.E-5&	-4.E-4	    \\
\multicolumn{1}{c}{}                    & $r=2\%$ &    2.E-4&	2.E-4&	2.E-4&	1.E-4&	2.E-5&	-1.E-4&	-3.E-4&	-1.E-3    \\
\multicolumn{1}{c}{}                    & $r=5\%$ &    4.E-4&	6.E-4&	4.E-4&	1.E-4&	-3.E-4&	-9.E-4&	-1.E-3&	-3.E-3   \\
\multicolumn{1}{c}{}                    & $r=7\%$ &    3.E-4&	5.E-4&	6.E-4&	7.E-5&	-6.E-4&	-3.E-3&	-2.E-3&	-4.E-3     \\
\hline

\multicolumn{1}{c}{\multirow{5}{*}{$p_5$}} & $r=0\%$ &  -2.E-6&	-8.E-7&	-1.E-6&	4.E-7&	-2.E-7&	-2.E-6&	-4.E-6&	-1.E-6    \\
\multicolumn{1}{c}{}                    & $r=1\%$ &     5.E-5&	3.E-5&	-2.E-5&	-6.E-5&	-2.E-4&	-3.E-4&	-4.E-4&	-9.E-4    \\
\multicolumn{1}{c}{}                    & $r=2\%$ &     1.E-4&	2.E-5&	-9.E-5&	-2.E-4&	-4.E-4&	-7.E-4&	-1.E-3&	-2.E-3   \\
\multicolumn{1}{c}{}                    & $r=5\%$ &     2.E-4&	-4.E-5&	-4.E-4&	-8.E-4&	-1.E-3&	-2.E-3&	-3.E-3&	-5.E-3  \\
\multicolumn{1}{c}{}                    & $r=7\%$ &     4.E-5&	-6.E-5&	-8.E-4&	-1.E-3&	-2.E-3&	-3.E-3&	-4.E-3&	-6.E-3    \\
\
\end{tabular}
\caption{De--Ameri\-cani\-za\-tion effects on pricing put options in the Merton model - Average error between the de-Americanized and European prices for each maturity.}
\label{Tab_Pricing_Merton_avg_errors}
\end{table}

\begin{table}[]\tiny
\centering
\begin{tabular}{lc|ccccccccc}
\multicolumn{1}{c}{}                    &   & $0.80$    & $0.85$    & $0.90$   & $0.95$   & $1.00$   & $1.05$   & $1.10$   & $1.15$ & $1.20$   \\ \hline
\multicolumn{1}{c}{\multirow{5}{*}{$p_1$}} & $r=0\%$&	8.E-05	&	1.E-04	&	2.E-04	&	3.E-04	&	3.E-04	&	3.E-04	&	2.E-04	&	2.E-04	&	3.E-04	\\
\multicolumn{1}{c}{}                    & $r=1\%$ &	8.E-05	&	1.E-04	&	2.E-04	&	3.E-04	&	3.E-04	&	3.E-04	&	3.E-04	&	3.E-04	&	2.E-04	\\
\multicolumn{1}{c}{}                    & $r=2\%$ &	7.E-05	&	1.E-04	&	2.E-04	&	3.E-04	&	3.E-04	&	3.E-04	&	3.E-04	&	2.E-04	&	2.E-04	\\
\multicolumn{1}{c}{}                    & $r=5\%$ &	5.E-05	&	9.E-05	&	2.E-04	&	2.E-04	&	2.E-04	&	2.E-04	&	3.E-05	&	-2.E-04	&	-6.E-04	\\
\multicolumn{1}{c}{}                    & $r=7\%$ &	3.E-05	&	7.E-05	&	1.E-04	&	2.E-04	&	2.E-04	&	8.E-05	&	-4.E-04	&	-8.E-04	&	0.E+00	\\
\hline
\multicolumn{1}{c}{\multirow{5}{*}{$p_2$}} & $r=0\%$&	2.E-05	&	4.E-05	&	7.E-05	&	1.E-04	&	2.E-04	&	2.E-04	&	2.E-04	&	2.E-04	&	1.E-04	\\
\multicolumn{1}{c}{}                    & $r=1\%$ &	3.E-05	&	6.E-05	&	1.E-04	&	2.E-04	&	3.E-04	&	4.E-04	&	4.E-04	&	4.E-04	&	4.E-04	\\
\multicolumn{1}{c}{}                    & $r=2\%$ &	4.E-05	&	7.E-05	&	1.E-04	&	2.E-04	&	3.E-04	&	4.E-04	&	4.E-04	&	4.E-04	&	3.E-04	\\
\multicolumn{1}{c}{}                    & $r=5\%$ &	5.E-05	&	1.E-04	&	2.E-04	&	3.E-04	&	4.E-04	&	3.E-04	&	-3.E-05	&	-4.E-04	&	0.E+00	\\
\multicolumn{1}{c}{}                    & $r=7\%$ &	6.E-05	&	1.E-04	&	2.E-04	&	3.E-04	&	4.E-04	&	8.E-05	&	-1.E-03	&	-1.E-03	&	0.E+00	\\
\hline
\multicolumn{1}{c}{\multirow{5}{*}{$p_3$}} & $r=0\%$&	2.E-05	&	3.E-05	&	4.E-05	&	5.E-05	&	8.E-05	&	1.E-04	&	1.E-04	&	9.E-05	&	7.E-05	\\
\multicolumn{1}{c}{}                    & $r=1\%$ &	9.E-05	&	1.E-04	&	2.E-04	&	2.E-04	&	3.E-04	&	4.E-04	&	4.E-04	&	4.E-04	&	5.E-04	\\
\multicolumn{1}{c}{}                    & $r=2\%$ &	2.E-04	&	2.E-04	&	3.E-04	&	4.E-04	&	5.E-04	&	6.E-04	&	6.E-04	&	7.E-04	&	8.E-04	\\
\multicolumn{1}{c}{}                    & $r=5\%$ &	3.E-04	&	5.E-04	&	6.E-04	&	8.E-04	&	1.E-03	&	1.E-03	&	1.E-03	&	1.E-03	&	1.E-03	\\
\multicolumn{1}{c}{}                    & $r=7\%$ &	5.E-04	&	6.E-04	&	8.E-04	&	1.E-03	&	1.E-03	&	1.E-03	&	2.E-03	&	1.E-03	&	2.E-03	\\
\hline
\multicolumn{1}{c}{\multirow{5}{*}{$p_4$}} & $r=0\%$&	-3.E-06	&	2.E-06	&	2.E-06	&	-3.E-06	&	3.E-06	&	-2.E-06	&	-3.E-07	&	-2.E-06	&	5.E-07	\\
\multicolumn{1}{c}{}                    & $r=1\%$ &	4.E-05	&	5.E-05	&	4.E-05	&	5.E-05	&	4.E-05	&	2.E-05	&	-9.E-06	&	-4.E-05	&	-9.E-05	\\
\multicolumn{1}{c}{}                    & $r=2\%$ &	1.E-04	&	1.E-04	&	2.E-04	&	2.E-04	&	2.E-04	&	1.E-04	&	8.E-05	&	5.E-06	&	-9.E-05	\\
\multicolumn{1}{c}{}                    & $r=5\%$ &	4.E-04	&	5.E-04	&	6.E-04	&	7.E-04	&	7.E-04	&	7.E-04	&	5.E-04	&	2.E-04	&	1.E-04	\\
\multicolumn{1}{c}{}                    & $r=7\%$ &	7.E-04	&	8.E-04	&	1.E-03	&	1.E-03	&	1.E-03	&	1.E-03	&	8.E-04	&	8.E-04	&	8.E-04	\\
\hline
\multicolumn{1}{c}{\multirow{5}{*}{$p_5$}} & $r=0\%$&	2.E-06	&	-2.E-07	&	-6.E-09	&	-3.E-07	&	4.E-06	&	1.E-06	&	2.E-06	&	7.E-08	&	3.E-06	\\
\multicolumn{1}{c}{}                    & $r=1\%$ &	1.E-04	&	2.E-04	&	2.E-04	&	2.E-04	&	2.E-04	&	3.E-04	&	3.E-04	&	3.E-04	&	3.E-04	\\
\multicolumn{1}{c}{}                    & $r=2\%$ &	3.E-04	&	4.E-04	&	4.E-04	&	5.E-04	&	5.E-04	&	6.E-04	&	6.E-04	&	7.E-04	&	7.E-04	\\
\multicolumn{1}{c}{}                    & $r=5\%$ &	9.E-04	&	1.E-03	&	1.E-03	&	1.E-03	&	2.E-03	&	2.E-03	&	2.E-03	&	2.E-03	&	2.E-03	\\
\multicolumn{1}{c}{}                    & $r=7\%$ &	1.E-03	&	1.E-03	&	2.E-03	&	2.E-03	&	2.E-03	&	2.E-03	&	3.E-03	&	3.E-03	&	3.E-03	\\
\hline
\end{tabular}
\caption{De--Ameri\-cani\-za\-tion effects on pricing put options in the Merton model - Average error between the de-Americanized and European prices for each strike. The empty fields are due to Remark \ref{Remark_Exercise_Region}.}
\label{Tab_Pricing_Merton_avg_errors_strikes}
\end{table}

\begin{table}[]\tiny
\centering
\begin{tabular}{lc|cccccccc}
                                        &   & $T_1$    & $T_2$    & $T_3$   & $T_4$   & $T_5$   & $T_6$   & $T_7$   & $T_8$   \\ \hline
\multirow{5}{*}{$p_1$}                     & $r=0\%$ &    0.200 & 0.200 & 0.201 & 0.203 & 0.207 & 0.214 & 0.221 & 0.248   \\
                                        & $r=1\%$ &      0.199 & 0.198 & 0.198 & 0.199 & 0.201 & 0.206 & 0.211 & 0.230     \\
                                        & $r=2\%$ &     0.198 & 0.196 & 0.195 & 0.195 & 0.196 & 0.199 & 0.202 & 0.212     \\
                                        & $r=5\%$ &    0.195 & 0.190 & 0.187 & 0.184 & 0.180 & 0.176 & 0.174 & 0.165    \\
                                        & $r=7\%$ &     0.193 & 0.186 & 0.181 & 0.177 & 0.170 & 0.163 & 0.157 & 0.138  \\
 \hline
\multirow{5}{*}{$p_2$}                     & $r=0\%$ &    0.200 & 0.200 & 0.200 & 0.201 & 0.203 & 0.208 & 0.214 & 0.237    \\
                                        & $r=1\%$ &     0.199 & 0.198 & 0.197 & 0.197 & 0.198 & 0.200 & 0.204 & 0.218     \\
                                        & $r=2\%$ &     0.198 & 0.196 & 0.194 & 0.193 & 0.192 & 0.192 & 0.194 & 0.200   \\
                                        & $r=5\%$ &    0.195 & 0.190 & 0.185 & 0.182 & 0.176 & 0.170 & 0.165 & 0.152  \\
                                        & $r=7\%$ &     0.193 & 0.186 & 0.180 & 0.174 & 0.165 & 0.155 & 0.148 & 0.125  \\
 \hline
\multirow{5}{*}{$p_3$}                     & $r=0\%$ &     0.200 & 0.201 & 0.205 & 0.210 & 0.221 & 0.237 & 0.252 & 0.302    \\
                                        & $r=1\%$ &      0.199 & 0.199 & 0.202 & 0.206 & 0.216 & 0.230 & 0.243 & 0.284    \\
                                        & $r=2\%$ &     0.198 & 0.197 & 0.200 & 0.203 & 0.211 & 0.223 & 0.234 & 0.268   \\
                                        & $r=5\%$ &     0.195 & 0.192 & 0.191 & 0.193 & 0.197 & 0.203 & 0.209 & 0.223    \\
                                        & $r=7\%$ &     0.193 & 0.188 & 0.186 & 0.186 & 0.188 & 0.191 & 0.194 & 0.197\\   
  \hline
\multirow{5}{*}{$p_4$}                     & $r=0\%$ &      0.205 & 0.212 & 0.221 & 0.232 & 0.256 & 0.288 & 0.315 & 0.398  \\
                                        & $r=1\%$ &      0.204 & 0.210 & 0.219 & 0.229 & 0.252 & 0.282 & 0.307 & 0.381    \\
                                        & $r=2\%$ &      0.203 & 0.208 & 0.216 & 0.226 & 0.247 & 0.276 & 0.299 & 0.365    \\
                                        & $r=5\%$ &      0.200 & 0.203 & 0.208 & 0.216 & 0.235 & 0.257 & 0.275 & 0.319   \\
                                        & $r=7\%$ &      0.198 & 0.199 & 0.203 & 0.210 & 0.227 & 0.246 & 0.260 & 0.292\\                                          
                                                                                  \hline
\multirow{5}{*}{$p_5$}                     & $r=0\%$ &    0.205 & 0.217 & 0.238 & 0.260 & 0.295 & 0.336 & 0.370 & 0.474   \\
                                        & $r=1\%$ &     0.204 & 0.215 & 0.236 & 0.258 & 0.291 & 0.330 & 0.362 & 0.457     \\
                                        & $r=2\%$ &     0.203 & 0.213 & 0.234 & 0.255 & 0.286 & 0.324 & 0.354 & 0.441   \\
                                        & $r=5\%$ &     0.200 & 0.208 & 0.228 & 0.247 & 0.274 & 0.306 & 0.330 & 0.394  \\
                                        & $r=7\%$ &     0.198 & 0.205 & 0.224 & 0.241 & 0.267 & 0.294 & 0.315 & 0.366\\
\hline										
\end{tabular}
\caption{De--Ameri\-cani\-za\-tion effects on pricing put options in the Merton model - Maximal European put prices.}
\label{Tab_Pricing_Merton_max_prices}
\end{table}
\FloatBarrier
\begin{itemize}
\item Olena Burkovska (burkovska@ma.tum.de)\\
Institute for Numerical Mathematics, Technische  Universit{\"a}t M{\"u}nchen,
85748 Garching b. M{\"u}nchen, Germany
\item Maximilian Ga{\ss} (maximilian.gass@tum.de)\\
Chair of Mathematical Finance, Technische Universit{\"a}t M{\"u}nchen, 85748
Garching b. M{\"u}nchen, Germany
\item Kathrin Glau (kathrin.glau@tum.de)\\
Chair of Mathematical Finance, Technische Universit{\"a}t M{\"u}nchen, 85748
Garching b. M{\"u}nchen, Germany
\item Mirco Mahlstedt (mirco.mahlstedt@tum.de)\\
Chair of Mathematical Finance, Technische Universit{\"a}t M{\"u}nchen, 85748
Garching b. M{\"u}nchen, Germany
\item Wim Schoutens (wim.schoutens@kuleuven.be) \\
Department of Mathematics, K.U.Leuven, Celestijnenlaan 200B   (box 2400), B-3001 Leuven, Belgium
\item Barbara Wohlmuth (wohlmuth@ma.tum.de)\\
Institute for Numerical Mathematics, Technische  Universit{\"a}t M{\"u}nchen,
85748 Garching b. M{\"u}nchen, Germany
\end{itemize}
\end{document}